\documentclass{article}
\usepackage{definitions}
\newtheorem{theorem}{Theorem}[section]

\newtheorem{lemma}[theorem]{Lemma}
\newtheorem{proposition}[theorem]{Proposition}

\newtheorem{definition}[theorem]{Definition}

\newtheorem{observation}[theorem]{Observation}

\newtheorem{prop}[theorem]{Proposition}

\newtheorem*{conjecture*}{Conjecture}
\newtheoremstyle{nonindented}{1ex}{1ex}{}{}{\bfseries}{.}{.5em}{}
\newtheoremstyle{indented}{1ex}{1ex}{\itshape\addtolength{\leftskip}{0.6cm}\addtolength{\rightskip}{0.6cm}}{}{\bfseries}{.}{.5em}{}
\theoremstyle{nonindented}

\theoremstyle{indented}
\newtheorem*{direction*}{Research Direction}
\theoremstyle{plain}

\newenvironment{alg}{\begin{algorithm}\begin{onehalfspace}\begin{algorithmic}[1]}{\end{algorithmic}\end{onehalfspace}\end{algorithm}} 

\newcommand{\1}[1]{\mathds{1}#1}

\renewcommand{\hat}{\widehat}
\renewcommand{\tilde}{\widetilde}
\renewcommand{\bar}{\overline}

\def\ex{\qopname\relax n{E}}
\def\min{\qopname\relax n{min}}
\def\max{\qopname\relax n{max}}

\def\argmax{\qopname\relax n{argmax}}

\newcommand{\jbid}{b^{\qopname\relax n{dual-opt}}_t}
\newcommand{\sbid}{b^{\qopname\relax n{seq}}_t}
\newcommand{\mbid}{b^{\qopname\relax n{min}}_t}

\newcommand{\RR}{\mathbb{R}}

\def\E{\mathcal{E}}

\def\R{\mathcal{R}}

\def\ex{\qopname\relax n{E}}

\def\min{\qopname\relax n{min}}
\def\max2{\qopname\relax n{max2}}
\def\max{\qopname\relax n{max}}

\def\argmax{\qopname\relax n{argmax}}

\newcommand{\opt}{\mathsf{OPT}}

\def\E{\mathcal{E}}

\def\R{\mathcal{R}}

\def\part{P} 
 
\def\v{\boldsymbol{v}}

\newcommand{\xt}{x_t}
\newcommand{\bt}{b_t}
\newcommand{\vt}{v_t}
\newcommand{\pt}{p_t}

\newcommand{\reg}{\mathrm{Regret}}
\newcommand{\rew}{\mathrm{Reward}}

\newcommand{\gseq}{\overrightarrow{\gamma}}

\global\long\def\1{\mathbf{1}}%

\global\long\def\vt{v_{t}}%
\global\long\def\xt{x_{t}}%
\global\long\def\pt{p_{t}}%

\global\long\def\bt{b_{t}}%
\global\long\def\0{\mathbf{0}}%
\global\long\def\ftscombined{f_{t}^{\star}}
\global\long\def\reg{\mathsf{Regret}}%
\global\long\def\rew{\mathsf{Reward}}%
\global\long\def\alg{\mathsf{Alg}}%
\global\long\def\mpacing{\mathsf{MinPacing}}%
\global\long\def\calp{\mathcal{P}}%
\global\long\def\E{\mathbb{E}}%
\newenvironment{lp*}{\begin{equation*}  \begin{array}{lll}}{\end{array}\end{equation*}}

\def\part{P} 
 
\def\v{\boldsymbol{v}}

\usepackage[capitalize, nameinlink]{cleveref} 
\usepackage{thmtools} 
\usepackage{thm-restate}

\crefname{prob}{Problem}{Problems}
\crefname{fact}{Fact}{Facts}
\crefname{alg}{Algorithm}{Algorithms}
\crefname{sec}{Section}{Sections}
\crefformat{none}{(#2#1#3)}
\crefname{equation}{Equation}{Equations}
\crefname{lem}{Lemma}{Lemmas}
\crefname{rem}{Remark}{Remarks}
\crefname{lemma}{Lemma}{Lemmas}
\crefname{defn}{Definition}{Definitions}
\crefname{cor}{Corollary}{Corollaries}
\crefname{prop}{Proposition}{Propositions}
\crefname{ineq}{Inequality}{Inequalities}
\crefname{assumption}{Assumption}{Assumptions}
\creflabelformat{ineq}{#2{\upshape(#1)}#3} 
\crefalias{thm}{theorem}
\let\ref\cref

\begin{document}
\date{}

\title{A Field Guide for Pacing Budget and ROS Constraints}

\author{Santiago R. Balseiro\thanks{Columbia Business School and Google Research (srb2155@columbia.edu).} \and Kshipra Bhawalkar\thanks{Google Research (kshipra@google.com)} \and Zhe Feng\thanks{Google Research (zhef@google.com)} \and Haihao Lu\thanks{The University of Chicago, Booth School of Business (haihao.lu@chicagobooth.edu). Part of the work was done at Google Research.} \and Vahab Mirrokni\thanks{Google Research (mirrokni@google.com)}\and Balasubramanian Sivan\thanks{Google Research (balusivan@google.com)} \and Di Wang\thanks{Google Research (wadi@google.com)}}

\maketitle

\begin{abstract}
Budget pacing is a popular service that has been offered by major internet advertising platforms since their inception. Budget pacing systems seek to optimize advertiser returns subject to budget constraints through smooth spending of advertiser budgets. In the past few years, autobidding products that provide real-time bidding as a service to advertisers have seen a prominent rise in adoption. A popular autobidding stategy is value maximization subject to return-on-spend (ROS) constraints. For historical or business reasons, the systems that govern these two services, namely budget pacing and ROS pacing, are not necessarily always a single unified and coordinated entity that optimizes a global objective subject to both constraints. The purpose of this work is to theoretically and empirically compare algorithms with different degrees of coordination between these two pacing systems. 

In particular, we compare (a) a fully-decoupled \emph{sequential algorithm} that first constructs the advertiser's ROS-pacing bid and then lowers that bid for budget pacing; (b) a minimally-coupled \emph{min-pacing algorithm} that runs these two services independently, obtains the bid multipliers from both of them and applies the minimum of the two multipliers as the effective multiplier; and (c) a \emph{fully-coupled} dual-based algorithm that optimally combines the dual variables from both the systems. Our main contribution is to theoretically analyze the min-pacing algorithm and show that it attains similar guarantees to the fully-coupled canonical dual-based algorithm. On the other hand, we show that the sequential algorithm, even though appealing by virtue of being fully decoupled, could badly violate the constraints. We validate our theoretical findings empirically by showing that the min-pacing algorithm performs almost as well as the canonical dual-based algorithm on a semi-synthetic dataset that was generated from a large online advertising platform's auction data. 
\end{abstract}

\section{Introduction}\label{sec:intro}
Internet advertisers purchase advertising opportunities by bidding in real-time auctions, and, to control their expenditures, it is common for advertisers to set budgets for their campaigns~\cite{GoogleBudgetPacing, FacebookBudgetPacing, TwitterBudgetPacing}. Budget pacing is a popular service offered by most advertising platforms that allows advertisers to specify their budgets and then optimizes advertiser bids in real-time to maximize advertisers' return subject to the spend being at most the budget. In the past few years, thanks to the increasing availability of ROS-related metrics, and the vastly improved conversion prediction models, autobidding products have seen a prominent rise in adoption~\cite{FacebookAutoBidding, GoogleAutoBidding}. These are tools that provide value-optimizing real-time bidding subject to return-on-spend (ROS) constraints (on top of the existing budget constraints) as a service to advertisers. Autobidding takes as input high-level advertiser goals like the target cost per conversion or acquisition of an advertiser and places real-time bids on a per-query basis to optimize advertiser returns. 

The algorithms that govern budget and ROS pacing, namely value-optimization subject to budget and ROS constraints, are not necessarily always a unified entity that optimizes a global objective. These services are often managed by different business units within the same organization or by different organizations altogether (many third-party demand-side platforms offer autobidding services), which results in different algorithms independently choosing/modifying advertisers' bids. This is not surprising in light of the meaningful gap between the times at which these products gained traction, with budget pacing systems having been standard and popular much earlier. As a result, even if the objectives of both services are aligned, the presence of budget and ROS constraints can introduce inefficiencies in the bidding process when the systems are decoupled. How do the fully decoupled and fully coupled optimal pacing services compare? Is there a way to operate the pacing service that obtains the best of both worlds: i.e., (a) maintain the theoretical guarantees of the fully coupled optimal pacing service, while (b) still being only minimally coupled? Our contribution in this work is to design and analyze an algorithm that approaches the best of both worlds. We establish this fact both theoretically and empirically.

\subsection{Pacing Services} 

Pacing services are online algorithms that adaptivelty adjust advertisers' bids based on auction feedback to maximize certain objectives while satisfying different constraint. Nowadays, a popular paradigm in internet advertising markets is that of \emph{value maximization}~\cite{FacebookAutoBidding, GoogleAutoBidding}. Unlike the usual quasilinear utility model, where the bidder seeks to maximize the difference between their value and payment, the bidder's stated objective in autobidding/budgeting products is to maximize their overall value (e.g., the number of conversions or conversion value) while respecting their budget and ROS constraints. For example, a bidder could ask to maximize the total number of conversions they get, subject to spending at most $\$1000$ and not paying more than $\$5$ per conversion. Figure~\ref{fig:joint} illustrates a \emph{joint optimization pacing service}, which we also refer to as a \emph{dual-optimal pacing service}, which takes as input the advertiser's budget and ROS target, and then automatically bids on behalf of the advertiser in the platform's auction. Importantly, the pacing services maintain a feedback loop that monitors the real-time spend and conversions from the auction and uses this information to adjust bids.

As we discussed, in many cases, the budget and pacing services maintain separate feedback loops. For historical reasons, budget pacing services are offered by platforms themselves, and ROS pacing services are built on top of them (they are either offered by the same advertising platform or third parties). In Figure~\ref{fig:sequential} we illustrate a typical \emph{sequential pacing service} in which the ROS pacing services feeds bids to the budget pacing service, which, in turn, bids in the platform's auction. Each service consumes the spend and conversion feedback from the auction to adjust bids dynamically. The benefit of the sequential optimization architecture is its decoupled nature, i.e., it could operate separate modules for budget pacing and ROS pacing.

We also consider a third minimally coupled architecture (Figure~\ref{fig:min}), which we call the \emph{min pacing service}. Rather than organizing the pacing services sequentially, they are organized in parallel. For each auction, the bid is obtained by taking the minimum of the bids generated by the two systems. While more generally one can think about other reduction operations of the two pacing systems' bids, as we show in this work, the min pacing already performs quite well and approaches the performance of the joint dual-optimal pacing service while still being only minimally coupled.

\begin{figure*}[htb!]
\centering
\begin{subfigure}[b]{0.33\textwidth}
\centering
\includegraphics[scale=0.4]{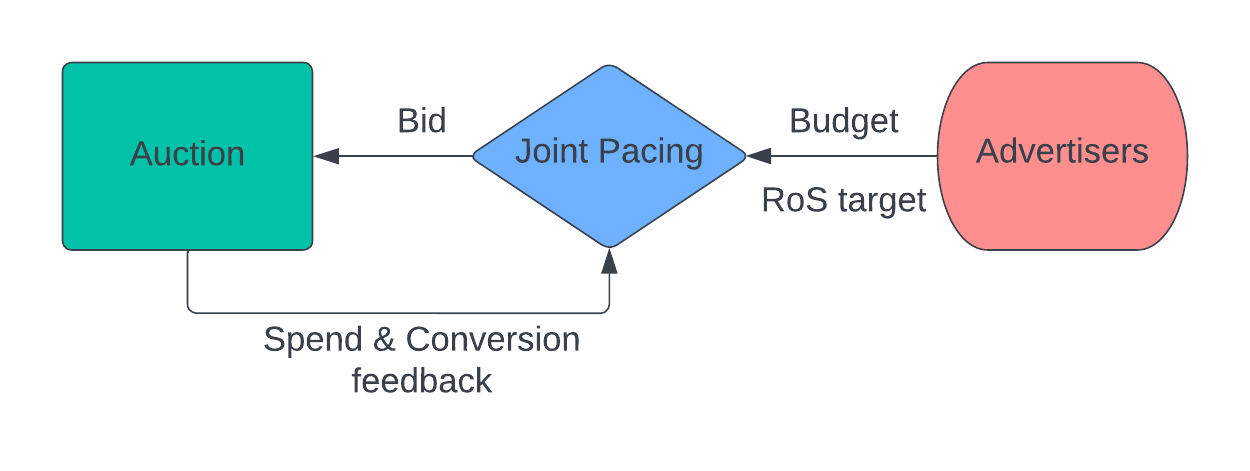}
\caption{Dual-Optimal Pacing Service}
\label{fig:joint}
\end{subfigure}
\begin{subfigure}[b]{0.33\textwidth}
\centering
\includegraphics[scale=0.4]{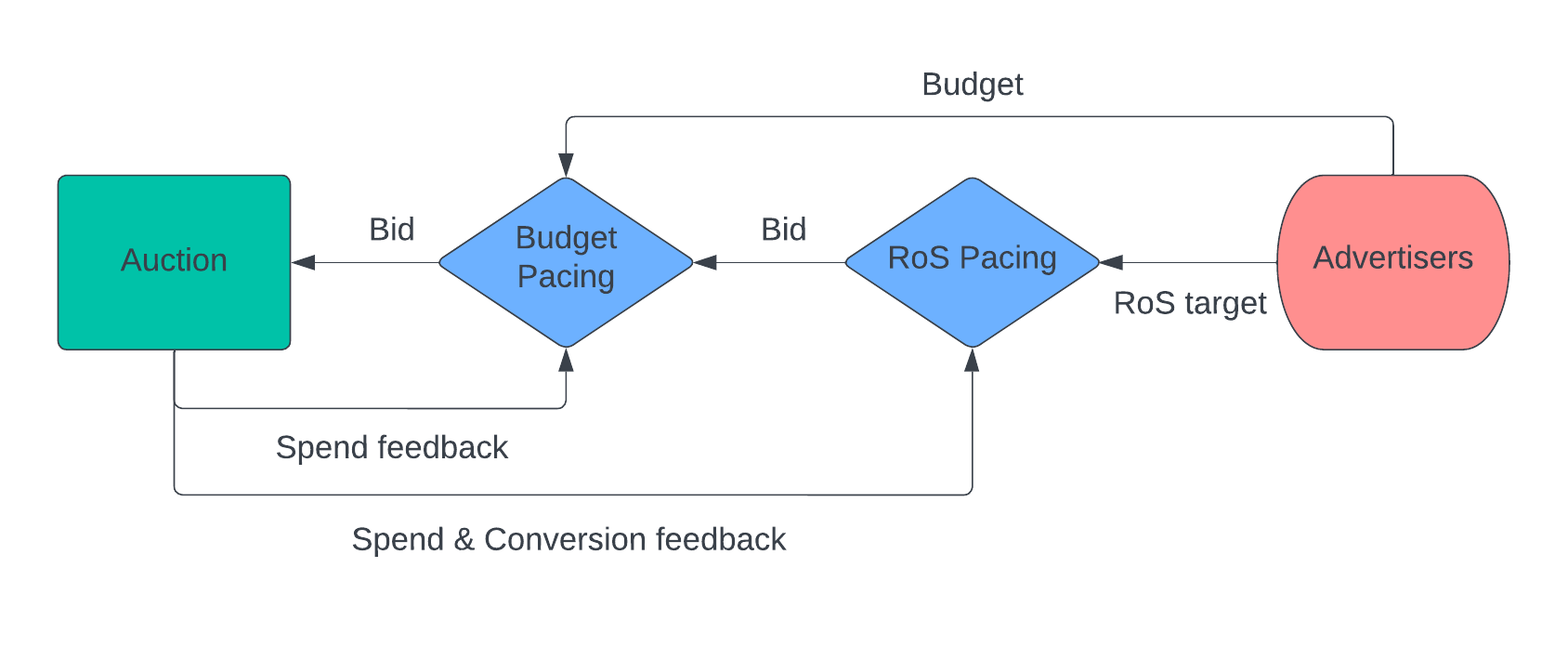}
\caption{Sequential Pacing Service}
\label{fig:sequential}
\end{subfigure}
\begin{subfigure}[b]{0.33\textwidth}
\centering
\includegraphics[scale=0.4]{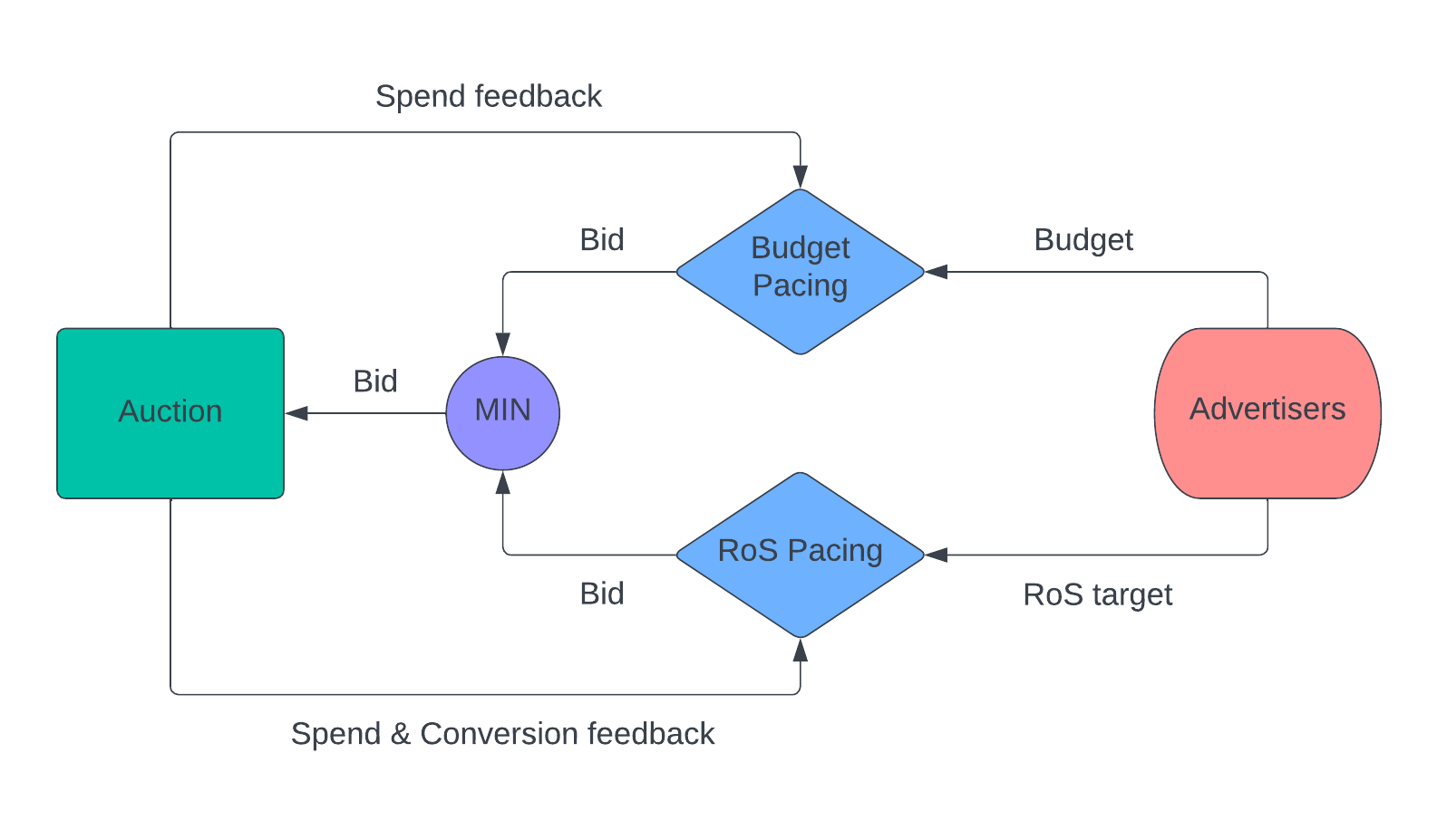}
\caption{Min Pacing Service}
\label{fig:min}
\end{subfigure}
\caption{Three Different Pacing Services for Budget and ROS Constraints}
\label{fig:three services}
\end{figure*}




\subsection{Our Results}

We compare all three algorithms described above, both theoretically and empirically. We next overview the algorithmic implementations of the pacing services, the empirical evaluation, and our theoretical analysis. Our main contribution is a theoretical analysis of the min-pacing algorithm and shows that it obtains the best of both worlds, i.e., its performance approaches that of the joint dual-optimal pacing service, while still being essentially decoupled much like the sequential pacing architecture. On the other hand, we show that the sequential architecture itself is a very poor choice: it either violates constraints by $\Omega(T)$ or has a regret of $O(T)$, when there is a finite horizon of $T$ repeated auctions. 

\paragraph{Algorithmic implementation.} In this work, we consider \emph{uniform bidding policies} (which were first proposed and analyzed in~\citep{FMPS07}) that multiplicatively scale advertisers' values, which are usually generated using advanced machine learning prediction algorithms~\citep{mcmahan2013ad,he2014practical,zhou2018deep,juan2016field,lu2017practical}. Uniform bidding is appealing for its simplicity, can be shown to be optimal in many settings, and is extensively used in practice~\citep{ABM19}. The bid multiplier $k$ of the uniform bidding policy is adjusted in real-time using a feedback loop. While many choices are possible for the feedback loop, in this work we consider Lagrangian dual algorithms, which are the work-horse algorithms of budget pacing~\citep{BM22}. At a high level, these algorithms introduce a dual variable for each constraint and then adjust these dual variables dynamically using a first-order algorithm. The final bid multiplier is calculated using these dual variables. Dual-based algorithms have strong performance guarantees and have been shown to subsume PID controllers---one of the most popular feedback controllers used in practice~\citep{tashman2020dynamic,zhang2016feedback,smirnov2016online,yang2019bid,ye2020cold,BLMS22}. Therefore, we believe the algorithms studied in this paper are representative of those used by pacing services in practice. We provide more details on the concrete algorithmic implementation in Section~\ref{sec:setup}.

\paragraph{Theoretical evaluation.} We evaluate the three algorithms along two dimensions: ROS constraint error, and conversion value. Budget constraints are hard in practice, i.e., pacing algorithms can no longer participate in auctions when budgets are exhausted. 
In contrast, ROS constraints are often soft: while small violations are permitted, large violations are undesirable. Finally, the conversion value garnered before the budget runs out should be as large as possible. We benchmark algorithms by looking at their regret relative to the conversion value of an offline optimum pacing strategy satisfying budget and ROS constraints. Our results are summarized in Table~\ref{tab:summary}.


\begin{table}[]
\small
\begin{center}
\begin{tabular}{|c|c|c|}
\hline
           & ROS Violation & Regret \\ \hline
Dual-Optimal      & $O(\sqrt T)$  & $O(\sqrt T)$ \\ \hline
Sequential & $\Omega(T)$   & $\Omega(T)$  \\ \hline
Min        & $O(\sqrt T)$  & $O(\sqrt T)$\\ \hline
\end{tabular}
\caption{Summary of our theoretical evaluation.}
\label{tab:summary}
\end{center}
\end{table}

\paragraph{Technical contribution.} Our evaluation is performed in a statistical environment under uncertainty. In other words, we assume that values and competing bids are drawn independently from a distribution that is unknown to the algorithms. We consider a finite horizon with $T$ repeated auctions in which the budget $B$ is proportional to the number of auctions, i.e., $B = \rho T$ for some fixed $\rho>0$. Recently, \cite{FPW22} showed that joint dual-optimal algorithm scores high along three dimensions. It runs out of budget at most $O(\sqrt{T})$ auctions from the end of the horizon, violates the ROS constraint by an amount $O(\sqrt{T})$, and attains a regret (conversion value relative to offline optimum) of $O(\sqrt{T})$. Our main result in this paper is to show that the min pacing algorithm also scores high along three dimensions, achieving  $O(\sqrt{T})$ bounds similar to those of the dual-optimal pacing algorithm. The analysis of the min pacing algorithm is challenging because we do not have access to a Lyapunov function, as in the dual-optimal pacing case. Instead, we analyze the algorithm by carefully studying the dynamics of the dual variables, which evolve according to a complex stochastic process. In particular, using the ODE technique for recursive algorithms, we first prove that the min-pacing algorithm quickly identifies which constraint binds and reaches the orbit of an optimal solution in $O(\sqrt T)$ time steps. Then, using stochastic stability tools, we argue that the algorithm never leaves the orbit of an optimal solution with high probability. We conclude by showing that the regret accumulated once the algorithm is in orbit is small using results from online convex optimization.

We finally argue that sequential pacing leads to unacceptable levels of ROS constraint violation or regret. In particular, we show any instantiation of the sequential pacing algorithm can have either linear (i.e. $\Omega(T)$) ROS violation or linear regret on some instances.  

\paragraph{Empirical evaluation.}
Section~\ref{sec:empirical} explains in detail our evaluation methodology, including how we construct our semi-synthetic dataset, how we obtain the different quantities in our optimization formulation~\eqref{eq:obj_x} based on real auction data. Here we give a high-level summary of our result. The objective of the algorithms is to maximize conversion value subject to budget and ROS constraints. In our simulations, as we explain in Section~\ref{sec:empirical} we enforce a hard stop once the budget constraint is violated, but we do not enforce a hard stop for the ROS constraint. This is aligned with practice as budget constraints are usually enforced more strictly than ROS constraints. As a result, we cannot compare conversion values directly because some algorithms might produce solutions that are infeasible, i.e., they could violate the ROS constraints. Therefore, we evaluate the different algorithms as follows. For each algorithm, we determine for each percentual level $z\%$ violation of the ROS constraint, the total conversion value obtained by the algorithm over all the campaigns that violated the constraint by at most $z\%$. By comparing these quantities, we can obtain the following critical insight: what percentage of ROS constraint violation does the naive sequential pacing need, to obtain the same value as the dual-optimal pacing does, at say $1\%$ constraint violation, or the min pacing does, at say $5\%$ constraint violation. Such plots are shown in Figure~\ref{fig:cumulative_value}. Similar plots, but instead focusing on the number of campaigns that violate the ROS constraint by $z\%$ is portrayed in Figure~\ref{fig:num_of_campaigns}. 

The high level summary is quite evident from these figures: \emph{the naive sequential pacing needs to violate the ROS constraint by a very significant percentage to approach anywhere near the dual-optimal pacing, while the min pacing approaches the dual optimal pacing at a much smaller percentage of ROS constraint violation. Moreover, in sequential pacing, the feedback loops of budget and ROS can lead to unstable dynamics.} Our findings suggest avoiding the sequential implementation despite its simplicity and appeal, and point towards having the two feedback loops either operating in a centralized manner, or at least minimally coupled as in the min pacing architecture.  Overall, our work has implications for the design and operation of pacing services. Our findings suggest that the lack of coordination of sequential pacing can lead to suboptimal and unstable outcomes. Advertising should, whenever possible, adopt algorithms that have some level of coordination between budget and ROS pacing. If centralized architecture is not an option, then the minimally-coupled min pacing architecture is a simple, practical and high-performant option to consider.

\subsection{Related Work}
We discuss here the paper that is most related to ours, and due to lack of space, we discuss other related work in Appendix~\ref{app:related}. In independent work, Lucier et al.~\cite{LPSZ23} studied a conceptually similar, yet different, algorithm that uses the final bid as the minimum of bid from the two pacing services. But unlike ours, the bid they use from each pacing service is different from the dual-optimal bid for that service. More importantly, they study a multi-bidder setting (unlike our single bidder setting) and their primary quantity of interest is the loss in liquid welfare, namely, the budget-capped sum of values obtained by all agents, when all of them employ this bidding algorithm. They establish that when the autobidding algorithms of agents play against each other, the resulting expected liquid welfare is at least half of the optimal expected liquid welfare achievable. But for the individual agent's regret, even in a single bidder stochastic setting where the competing bids are drawn i.i.d. (rather than being set by other players simultaneously adopting the same algorithm), their regret bound is $O(T^{7/8})$ as opposed to the tight $O(\sqrt{T})$ guarantee we prove. In general, proving a low regret guarantee when all bidders are simultaneously using the same algorithm requires strong assumptions that guarantee that the optimal dual variables of all agents converge to a unique Nash equilibrium. Such an analysis is provided by Balseiro ang Gur~\cite{BG19} for the case of utility-maximizing budget-constrained bidders under a strong-monotonicity assumption of the bidders expenditures. We conjecture our analysis could be extended to show similar low-regret guarantees in multi-bidder settings under similar strong assumptions, and we leave this is as an interesting research direction.

\section{The setup}\label{sec:setup}

In this section, we define a formal model for budget and ROS constraint pacing. We consider a single bidder who participates in $T$ repeated auctions. The bidder derives a value of $v_t\in [0,1]$ from getting allocated in auction $t = 1,\ldots,T$. Upon submitting a bid of $b_t$, the bidder gets an allocation of $x_t(b_t)$ and an expected payment of $p_t(b_t)$. I.e., $x_t: \R_{\geq 0} \rightarrow [0,1]$, and $p_t: \R_{\geq 0} \rightarrow [0,1]$ are the allocation and payment functions respectively. Note we assume without loss of generality (by scaling) that $v_t,x_t,p_t$ are all in $[0,1]$. The tuple $\gamma_t=(v_t, x_t, p_t)$ is drawn i.i.d. every round from an unknown distribution $\calp$. We denote the sequence of $T$ samples by $\gseq := \{\gamma_1, \gamma_2, \dots, \gamma_T\} \sim \calp^T$ and sequences of length $\ell\neq T$ by $\gseq_{\!\ell}$ where needed.
At the beginning of round $t$, the bidder has knowledge of the value $v_t$ and the historical information of past auctions to decide on a bid, $b_t$. Denote $\delta_t = (x_t(b_t), p_t(b_t))$ to represent the outcome of the auction at round $t$. At the end of round $t$, the bidder observes $\delta_t$. Thus the historical information at the beginning of round $t$ is $h_t = \{(v_s, \delta_s)\}_{s\leq t-1}$. 

\paragraph{The optimization objective} 
The advertiser is a \emph{value-maximizer} and seeks to maximize the overall value while respecting the budget of $B$ dollars and the ROS constraint.  Formally, the bidder's optimization problem is stated as follows:
\begin{equation}
\label{eq:obj_x}
\begin{array}{ll}
\underset{b_{t}: t=1, \cdots, T}{\mbox{maximize}} & \sum_{t=1}^T v_{t}\cdot x_t(b_t)\\
\mbox{subject to } &  \sum_{t=1}^T p_t(b_t) \leq  \sum_{t=1}^T v_{t} \cdot x_t(b_t),
\\
& \sum_{t=1}^T p_t(b_t) \leq B
\end{array}
\end{equation}
The first constraint is the ROS constraint, which states that for every dollar spent, there is at least a dollar of value.\footnote{More generally, one can have the constraint to state that for every dollar spent, there is at least $\tau$ dollars of value. But without loss of generality, one can set $\tau=1$. The update to the bidding formula as a function of $\tau$ is quite straightforward, and we skip this here to avoid carrying the notational clutter of $\tau$ everywhere.} The second constraint is the budget constraint. We define the per-round budget by $\rho:=B/T$. In round $t$ the bidder bids $b_t = \pi_t(v_t, h_t)$. The function $\pi_t(\cdot, \cdot)$ could be randomized. 


\paragraph{Truthful auctions, nontruthful auctions, uniform bidding policy.} 
{We restrict attention to a uniform bidding policy, i.e., 
one computes a bid multiplier $k_t$ independently of the current value $v_t$, and the bid submitted is $b_t = k_t\cdot v_t$.} If the underlying auction is truthful 
\footnote{An auction is truthful if the allocation function $x_t(b_t)$ is weakly monotonically increasing, and the payment function satisfies $p_t(b_t) = p_t(0) +  b_tx_t(b_t) - \int_0^{b_t}x_t(z)dz.$ In truthful auctions, quasi-linear utility maximizers are willing to report their value truthfully.}, Aggarwal et al.~\cite{ABM19} showed that the optimal bidding algorithm for problem~\eqref{eq:obj_x} is indeed a uniform bidding policy, and hence the restriction to uniform bidding is without loss of generality. If the underlying auction is non-truthful, the restriction to uniform bidding can be made without loss if the buyer has access to an optimizer $g_t(v)$ that computes the optimal bid to submit in a one-shot auction for any given true value\footnote{If the bidder had access to $x_t(\cdot)$ and $p_t(\cdot)$ before placing the bid at time $t$, the optimizer is $g_t(v) \in \argmax_b \left\{ v \cdot x_t(b) - p_t(b) \right\}$.} $v$. In this case, bidding $b_t = g_t(k_t \cdot v)$ would be optimal for the bidder due to the revelation principle.
\vspace{-5pt}
\subsection{The Bidding  Algorithms}\label{sec:algorithms} 
Despite the simplicity and appeal of uniform bidding, computing the optimal multiplier $k_t$ 
requires knowledge of the entire set of $\{v_t, x_t, p_t\}_{t=1\dots T}$, while information is only revealed in an online manner. 
Thus, to approach the performance of uniform bidding policy in an online setting,
a standard technique is to dualize the constraints and look at the Lagrangian dual of the problem. 
We introduce dual variables $\mu \ge 0$ for the budget constraint and $\lambda$ for the ROS constraint and write Lagrangian dual of the problem~\eqref{eq:obj_x}:
\begin{align}
\label{eq:Langrangian}
\underset{\lambda \ge0, \mu\ge 0}{\mbox{min}} ~~ ~~\underset{b_{t}: t=1, \cdots, T}{\mbox{max}} \left\{ T \rho \mu + \sum_{t=1}^T \Big( (1+\lambda) v_{t}\cdot x_t(b_t) - (\mu+\lambda) p_t(b_t) \Big) \right\}\ .
\end{align}

At each time $t$, the Lagrangian dual variables $\lambda_t, \mu_t$ can be updated using online mirror descent, which is a generic framework such that different instantiations can lead to distinct dual update steps and theoretical guarantees.
Then we compute the multiplier $k_t$ as a function of 
$\lambda_t, \mu_t$ and set the bid of $b_t = k_t\cdot v_t$. 

{
\paragraph{Dual-Optimal Pacing}
We now discuss how to derive the optimal bidding multiplier $k_t$ when the underlying auction is truthful. Note that the Lagrangian dual problem \eqref{eq:Langrangian} becomes separable over time after dualizing the constraints. Therefore, at time $t$, assuming that the dual variables are $\mu_t$ and $\lambda_t$, the optimal bid by solving \eqref{eq:Langrangian} is
}

\begin{align}
\jbid &= \argmax_{b} \left\{ (1+\lambda_t) \cdot v_{t}\cdot x_t(b) - (\mu_t+\lambda_t) p_t(b) \right\} \nonumber \\
&= \argmax_{b} \left\{ \frac{1+\lambda_t}{\mu_t+\lambda_t} \cdot v_{t}\cdot x_t(b) -  p_t(b) \right\} 
= \frac{1+\lambda_t}{\mu_t + \lambda_t} \cdot v_t\,,\label{eq:joint-bid}
\end{align}
where the second equation follows from extracting the factor $(\mu_t+\lambda_t)$ and the last because the bidder's problem is equivalent to that of bidding in a truthful auction when the value is $(1+\lambda_t) / (\mu_t + \lambda_t) v_t$. In other words $k_t = (1+\lambda_t)/(\mu_t + \lambda_t)$. Note that $k_t$ is multiplicatively \emph{inseparable} across $\lambda_t$ and $\mu_t$, {therefore, we need a centralized pacing to update $k_t$}.

The dual variables are updated using feedback loops based on the auction result that have natural self-correcting features to prevent constraint violations (see Algorithm~\ref{alg:joint-updates-mirror-descent}). For example, in the case of the budget constraint, the feedback loop in \eqref{eq:budget-dual-update} seeks to equate the actual spend of the auction $p_t(b_t)$ with the per-round budget $\rho$ to satisfy the budget constraint (whenever this constraint is binding). 
Mathematically, the dual variable updates (i.e. the mirror descent step) are derived from solving a dual optimization problem, and we {apply exponential updates for both budget and ROS dual variables}
in \eqref{eq:ros-dual-update} and 
\eqref{eq:budget-dual-update}\footnote{These exponential update rules are derived by instantiating with a particular Bregman divergence function in the mirror descent setup.}. 
We refer the reader to \cite{BG19,FPW22} for more details. \citet{FPW22} show that this specific setup obtains near-optimal regret $O(\sqrt{T})$, where regret is the difference between the offline optimal total value and the bidding policy's total value. 

\begin{algorithm}
\caption{Dual-Optimal Pacing 
}\label{alg:joint-updates-mirror-descent}
\textbf{Initialize:} Initial dual variables $\lambda_1=1$, $\mu_1=0$, total initial budget $B_1 := \rho T$,  gradient descent step-sizes $\alpha$ and $\eta$;

\For{$t=1,2,\cdots, T$}
{ Observe the value $v_t$, and set the bid \[b_t = \min\left\{\frac{1+  \lambda_t}{\mu_t + \lambda_t}\cdot v_t, B_t\right\}.\]


Update the dual variable of the ROS constraint 
\begin{eqnarray}\label{eq:ros-dual-update}
\lambda_{t+1} := \lambda_t \cdot \exp\Big(-\alpha \cdot \left( v_t \cdot x_t({b}_t) - p_t(b_t) \right) \Big)\,.
\end{eqnarray}%


Update the dual variable of the budget constraint as 
\begin{eqnarray}\label{eq:budget-dual-update}
\mu_{t+1} := 
\mu_t \cdot \exp\Big( - \eta \cdot \left(\rho - p_t(b_t)\right)\Big)\,.
\end{eqnarray}
    
Update the leftover budget $B_{t+1} = B_t - p_t(b_t)$;
}
\label{alg:joint-pacing}
\end{algorithm}

\paragraph{Sequential Pacing.} 
If one were to consider the problem~\eqref{eq:obj_x} with just the budget constraint, the bidding policy (from Lagrangian duality with the ROS dual variable $\lambda_t = 0$) would be to bid $b_t = v_t/\mu_t$, with the dual variable $\mu_t$ alone getting updated as in Algorithm~\ref{alg:joint-updates-mirror-descent}. Similarly, if one were to consider the problem ~\eqref{eq:obj_x} with just the ROS constraint, the bidding policy (from Lagrangian duality with budget dual variable $\mu_t = 0$) would be to bid $b_t = v_t\cdot\frac{1+\lambda_t}{\lambda_t}$, with the dual variable $\lambda_t$ alone updated as in Algorithm~\ref{alg:joint-updates-mirror-descent}. Given the historical context mentioned earlier, budget pacing systems have been around for longer than ROS pacing optimization. Therefore, it is not unexpected to have separate servers handling the feedback loops of the budget and ROS constraints and the final bid constructed in a sequential manner, namely, 
\begin{eqnarray}\label{eq:seq-update}
\sbid = \min\left\{\frac{1 + \lambda_t}{\lambda_t} \cdot \frac{1}{\mu_t} \cdot v_t,B_t\right\}.
\end{eqnarray}
In other words, the ROS constraint pacing service determines an intermediary bid  $\hat{b}_t = (1 + \lambda_t)/\lambda_t \cdot v_t$ which is fed to the budget service and, in turn, the budget pacing service operates on the scaled bid $\hat{b}_t$ to get the final bid of $\hat{b}_t/\mu_t$ (and also capped by the remaining budget $B_t$). While not optimal, this implementation has the benefit of being decentralized, i.e., it could operate separate servers for budget pacing and ROS pacing, that (a) only communicate the temporary bid $\hat{b}_t$ and (b) could update their respective variables at different frequencies.

\paragraph{Min Pacing.}
If the transition from sequential to dual-optimal pacing proves prohibitively expensive in the short term for organizational or engineering reasons, we propose and study another decentralized optimization, that we call the \emph{min pacing} service. Rather than applying the bid-lowering operations of the two pacing systems sequentially, we take the minimum of the bids generated by the two systems:
\begin{eqnarray}\label{eq:min-update}
\mbid = \min\left\{\frac{1 + \lambda_t}{\lambda_t}\cdot  v_t, \frac{1}{\mu_t} \cdot v_t, B_t\right\}\,.
\end{eqnarray}
The corresponding dual variables can follow the same update rules in Algorithm~\ref{alg:joint-updates-mirror-descent} (\eqref{eq:ros-dual-update} and \eqref{eq:budget-dual-update}). The min pacing service operates in parallel instead of sequentially and also requires minimum coordination between budgeting and ROS pacing. We will show that even though $\mbid$ is in general different from $\jbid$, and thus not the optimizer of the Lagrangian dual \eqref{eq:Langrangian}, bidding $\mbid$ nonetheless achieves the same asymptotically optimal guarantees on regret and constraint violation as the dual-optimal pacing algorithm.
\section{Theoretical Analysis of the Bidding Algorithms}
In this section we analyze the performances of the pacing algorithms introduced in the previous section, and we use the notions of regret and constraint violation. To define the regret, we first define the reward of some pacing algorithm $\alg$ for a sequence of requests $\gseq$ over a time horizon $T$ as \[ \rew(\alg, \gseq):=\sum_{t = 1}^T \vt\cdot \xt(b_t), \numberthis\label{eq:defReward}\] where $b_t$'s are the algorithm's bids. Note the definition doesn't require $\alg$ to satisfy the budget and ROS constraints. Next, we define the optimal reward $\rew(\opt,\gseq)$ for a sequence $\gseq$ as the optimal objective of the offline optimization problem~\eqref{eq:obj_x} given $\gseq$. The regret of $\alg$ is \[\reg(\alg, \calp^T):= \E_{\gseq\sim \calp^T} \left[ \rew(\opt,\gseq) - \rew(\alg, \gseq)\right].\numberthis\label{eq:defRegret}\] We remark that we define $\rew$ for some specific drawn sequence, whereas $\reg$ is defined with respect to a distribution. 

{
Note that $\reg$ itself does not fully measure the performance of $\alg$ since the reward of $\alg$ does not capture the budget and ROS constraints. All the pacing algorithms we discuss will cap the bid by the remaining budget, so the budget constraint is always satisfied. To evaluate our algorithms, we first need the following notion of stopping time.
\begin{definition}\label[defn]{defn:Defs}
The stopping time $\tau$ of \cref{alg:joint-pacing}, with budget $B$ is the first time $\tau$ at which $\sum_{t = 1}^{\tau} \pt(\bt) + 1 \geq B.$ 
\end{definition}

Intutively, $\tau$ is the first time step when the total payment almost exceeds the total budget. By budget endurance, we mean that $T - \tau$ is small for any $\gseq$ or, in other words, the budget always runs out close to the end of the horizon.}

In addition, we focus on the violation of the ROS constraint, i.e. $\sum_{t=1}^T p_t(b_t) - \sum_{t=1}^T v_{t} \cdot x_t(b_t)$. For constraint error we look at ex-post guarantees that hold for any $\gseq$.

In particular, both the joint pacing and min pacing algorithms achieve asymptotically nearly optimal guarantees in terms of the regret and constraint error in the stochastic i.i.d. setting. For simplicity we assume in our analysis that the allocation and payment functions are from truthful auctions. The result for the joint algorithm is already known from previous work in~\cite{BLM22,FPW22}. 

We start with analyzing the regret of $\mpacing$ by considering a continuous-time approximation of the algorithm in which multipliers are updated using the expected gradients instead of their noisy stochastic counterparts used in the real algorithm. 

Before proceeding with our analysis we provide some useful definitions. We denote by 
\begin{align*}
    \gb(k) = \rho - \E_{v} [p(k \cdot v)] \ ,~
    \gr(k) =\E_{v}[ v \cdot x(k \cdot v) - p(k \cdot v)]
\end{align*}
the expected error in the budget and ROS constraints when the multiplier is $k$. We plot some examples in Figure~\ref{fig:gradients}, which is located in the appendix. We require the following assumptions in our analysis.

Our first assumption is that the functions $\gr$ and $\gb$ cross zero once and from above, and that they are Lipschitz continuous.

\begin{ass}\label{ass:single-crossing}We assume that the functions $\gr(k)$ and $\gb(k)$ are $L_g$-Lipschitz continuous in $k$ and bounded. Moreover, the following hold: 

(1) The function $\gb(k)$ crosses the non-negative $k$-axis once at $\kb>0$ and from above. That is, for any $0\le k<\kb$, we have $\gb(k) > 0$ and for any $k>\kb$, we have $\gb(k)<0$.
            
(2) The function $\gr(k)$ crosses the positive $k-$axis once at $\kr>1$ and from above. That is, for any $0<k<\kr$, we have $\gr(k)>0$ and for any $k>\kr$, we have $\gr(k)< 0$.
\end{ass}

{
When the auction is truthful, it can be shown that the functions $\gr$ and $\gb$ always cross the positive axis from above. Therefore, the set of crossing points is always an interval. Assumption~\ref{ass:single-crossing} rules out the possibility of multiple crossing points and, as we shall discuss later, implies the uniqueness of the optimal bidding strategy. This assumption is related to the so-called ``general position'' condition, which is pervasive in online allocation problems (see, e.g., \citealt{DH09,AWY14}). The Lipschitz continuity of the gradients is a common assumption in the analysis of online algorithms~\citep{hazan2016introduction} and holds when either the interim allocation and payment are smooth, or the distribution of values is absolutely continuous. For example, this assumption might fail to hold in a second-price auction when values and competing bids are discrete (there, $\gr$ and $\gb$ are piecewise constant). In this case, it is possible to recover Lipschitz continuity by adding a small amount of random noise to the bids, which mollifies the functions $\gr$ and $\gb$, without significantly impacting the performance of our algorithm. As a result, Assumption~\ref{ass:single-crossing} is not too restrictive.}

Under Assumption~\ref{ass:single-crossing}, we can upper bound the optimal performance in terms of the value collected by a uniform bidding policy that bids the minimum of the multipliers $\kb$ and $\kr$, and provide a simple characterization of an optimal dual solution. The dual problem becomes $\min_{\mu \ge 0, \lambda \ge 0} D(\mu, \lambda)$ where 
{\footnotesize
\begin{align*}
    D(\mu, \lambda):= \underset{k\ge0} {\mbox{maximize}} \left\{ (1+\lambda) \E_{v}[ v \cdot x(k \cdot v) ] + \rho\cdot \mu- (\lambda+\mu) \E_{v} [p(k \cdot v)]\right\}\,,
\end{align*}}
is the dual function.

\begin{lemma}\label{lemma:optimal-multipler} Suppose Assumption~\ref{ass:single-crossing} holds. There exists an optimal solution with $\lambda^*=0$ and $\mu^* = 1/\kb$ if $\kb \le \kr$ or $\lambda^*=1/(\kr - 1)$ and $\mu^* = 0$ if $\kr \le \kb$. Moreover, we have that 
\[
    \E_{\gseq\sim \calp^T} \left[ \rew(\opt,\gseq) \right] \le T \cdot D(\mu^*,\lambda^*) = T \cdot \E_{v}\left[ v \cdot x(k^* \cdot v)\right]\,,
\]
where $k^* = \min(\kr,\kb)$.
\end{lemma}


\begin{ass}\label{ass:nondegenerate}
    The problem is non-degenerate, namely, $\kb\not=\kr$.
\end{ass}

The non-degeneracy assumption guarantees that only one of the budget constraint or the ROS constraint can be binding for the uniform bidding policy. In practice, the data comes from a random process, and the budget and targets are given by the advertiser. Notice that the degenerate case stays in a lower dimension manifold, thus it is very likely that the non-degenerate assumption holds. Under the non-degeneracy assumption, the optimal multiplier is either $k^* = \kb$ or $k^*=\kr$. { We remark that this assumption is only imposed to simplify the analysis---we can provide a similar regret bound of $\sqrt{T}$ for the degenerate case using techniques similar to the ones presented in this paper.}

\begin{ass}\label{ass:second-moment}The gradients of the budget and ROS constraints have second moments bounded by $\bar {G_2}$.
\end{ass}

{
Assumption~\ref{ass:second-moment} is common in the analysis of first-order algorithms for online optimization, where second moments are usually required to be bounded~\citep{hazan2016introduction}. When the auction is truthful, a sufficient condition for this assumption to hold is that values have bounded second moments, i.e., $\E_{v} [v^2] < \infty$.}

\begin{ass}\label{ass:strongly-monotone} There exists some $\delta > 0$ and $\ell > 0$ such that for all $k \in [k^* - \delta, k^*+\delta]$ we have that either $-\gb(k) (k - k^*) \ge \ell \cdot (k - k^*)^2$ if $k^* = \kb$ or $-\gr(k)(k-k^*) \ge \ell \cdot (k - k^*)^2$ if $k^* = \kr$.  
\end{ass}

{Our final assumption requires that the spend and conversion value are locally strongly monotone around the optimal solution. In other words, we require the gradients to be locally linear around point where they cross the positive axis, for example, when the budget constraint is binding, we require that $\gb(k) \approx k^* - k$ around $k^*$ (see Fig~\ref{fig:gradients}). Similarly to our Lipschtiz condition, we can guarantee this assumption holds by randomly perturbing bids. Assumption~\ref{ass:strongly-monotone} is also common in the analysis of online algorithms, where it is sometimes assumed that objective functions are strongly convex, which is equivalent to gradients being strongly monotone~\citep{hazan2016introduction}.}

{
The next theorem shows that the $\mpacing$ algorithm has an $O\left(T^{1/2}\right)$ regret bound.
\begin{thm}\label{thm:min_regret}
    Suppose Assumption \ref{ass:single-crossing}-\ref{ass:strongly-monotone} hold. Then, the regret of $\mpacing$ can be bounded as:
    \[
    \reg(\mpacing,\calp^T) = O\left(T^{1/2}\right)\,.
    \]
\end{thm}
One can show that the optimal joint algorithm also has $O\left(T^{1/2}\right)$ regret bound, which showcases that the $\mpacing$ algorithm achieves good practical performance.
}

Analyzing the min algorithm is challenging as we do not have access to a Lyapunov function as in the joint pacing case. We analyze the algorithm by carefully studying the dynamics of the dual variables under the min algorithm. Note that in light of Lemma~\ref{lemma:optimal-multipler}, if we knew in advance which constraint is binding, then we could attain low regret by bidding using the multiplier associated with the binding constraint ($k=1/\mu$ for the budget constraint and $k = (1+\lambda)/\lambda$ for the ROS constraint). Our proof technique is to show that with high probability the algorithm detects in $\sqrt{T}$ steps which constraint is binding and then bids according to the optimal bidding multiplier for the binding constraint.

To do so, we consider a continuous time approximation of the multipliers $(\bar\lambda(s)), \bar \mu(s))$ in which we update them continuously according to the expected gradients, and dynamics are governed by an ODE. The ODE traces the ``expected'' path of the multipliers when the step-size is small. Here, we assume that step-sizes are $\alpha>0$ for both constraints. The ODEs are obtained by considering the continuous approximation of multiplicative weight updates:
\begin{align}
    \frac d {ds} \log(\bar \lambda(s)) &= - \gr\left( k^{\min}(\bar \lambda(s), \bar \mu(s) \right)\,,\label{eq:ode}\\
    \frac d {ds} \log(\bar \mu(s)) &= - \gb\left( k^{\min}(\bar \lambda(s), \bar \mu(s) \right)\,,    \nonumber
\end{align}
where 
\[
    k^{\min}(\lambda,\mu) = \min\left( \frac{1+\lambda} \lambda , \frac 1 {\mu}\right)\,.
\]
The time in the ODE, which is denoted by $s \ge 0$ can be mapped to a step $t$ in the discrete-time stochastic system by setting $s = \alpha t$. In other words, the time in the ODE corresponds to the total distance traveled according to the step-size. We assume throughout that the ROS constraint binds at optimality. A similar analysis holds for the budget constraint. Our proof strategy is the following.

(1) \emph{Binding Constraint Identification.} Setting the step-size to be $\alpha \approx T^{-1/2}$, we show it takes order $\sqrt{T}$ steps to be get to an orbit of size $\epsilon$ of an dual optimal solution $(\lambda^*,0)$. The orbit is chosen so that the bidding formula in this region is $k = (1+\lambda)/\lambda$. We do so by first arguing that the ODE gets in a constant amount of time to the orbit of the optimal and then arguing that the actual algorithm remains close to the expected path traced by the ODE with probability $T^{-1/2}$ using a discrete version of Gronwall’s Lemma to bound the absolute deviations and then invoking a concentration argument to bound the maximum deviation in a stochastic sense.

(2) \emph{Orbital Stability.} Once the algorithm reaches an orbit of an optimal solution, we show that it never leaves the orbit with probability $T^{-1/2}$. {We prove this result by constructing a local stochastic Lyaponuv function using the KL divergence and then invoking a classical result from stochastic stability.}

(3) \emph{Regret Analysis.} We conclude by showing that the regret accumulated once the algorithm is in the orbit of the optimal solution is $\sqrt{T}$. {For this step, we first lower bound the conversion value collected by the algorithm in terms of the dual function $D(\lambda, \mu)$ and a complementary slackness term. Using weak duality, we can relate the first term to the optimal performance. The complementary slackness term is controlled using standard regret bounds for multiplicative weight updates.}
{ A detailed proof of Theorem \ref{thm:min_regret} is presented in Appendix \ref{app:min_proof}.}



{
Next, we present the constraint violation of the MinPacing algorithm. Recall that the bid is capped by the remaining budget; thus, the budget constraint will always be satisfied. Instead, we show that the budget always runs out $O(\sqrt{T})$ close to the horizon's end. On the other hand, for ROS constraint, we allow small violations throughout the horizon, and we can show that the violation is at most $O(\sqrt{T})$.

\begin{thm}\label{thm:constaints}
    Suppose payments are at most the bid, i.e., $p_t(b) \le b \cdot x_t(b)$ for all $b \ge 0$. Then, $\mpacing$ satisfies the following:
    
    \textbf{(ROS constraint.)} The violation of the ROS constraint is at most $O(\sqrt{T}\log T)$, i.e.,
        $\sum_{t=1}^T p_t(b_t) -v_t \cdot x_t({b}_t) \leq O(\sqrt{T}\log T).$
   
    \textbf{(Budget endurance.)} The budget always runs out close to the end of the horizon, i.e., stopping time $\tau$ satisfies $T-\tau\le O(\sqrt{T})$.
\end{thm}
The proof of Theorem \ref{thm:constaints} can be found in Appendix \ref{app:constraints}. 

Finally, we show that the sequential algorithm fails to work -- it may have $\Omega(T)$ regret and/or $\Omega(T)$ constraint violation, thus it is not a sub-optimal algorithm. The proof of Proposition \ref{prop:sequencial} is presented in Appendix \ref{app:sequential_fails}.

\begin{proposition}\label{prop:sequencial}
    For any initial dual variables $\mu_0,\lambda_0$ and step-sizes $\eta,\alpha$, there is an instance on which the algorithm either violates the ROS constraint by at least $\Omega(T)$ or has a regret at least $\Omega(T)$.
\end{proposition}
}

\section{Empirical Study}\label{sec:empirical}




\begin{table*}[h]
\caption{Cumulative fraction of campaigns and total conversion (normalized by total benchmark) over the ROS relative error buckets.} 
\centering
 \tiny
\begin{tabular}{cccccccccccccc}
\toprule
 & & \multicolumn{12}{c}{Relative Constraint Violation}\\
 & Alg &  $(\le)0$ & $0.05$ & $0.10$ & $0.15$ & $0.20$ & $0.25$ & $0.30$ & $0.35$ & $0.40$ & $0.45$ & $0.50$ & $\infty$ \\
\midrule 


Frac. & Dual-optimal & $0.62$ & $0.71$ & $0.75$ & $0.78$ & $0.80$ & $0.82$ & $0.83$ & $0.84$ & $0.85$ & $0.86$ & $0.87$ & $1.00$  \\
of & Min.
& $0.49$ & $0.64$ & $0.72$ & $0.77$ & $0.80$ & $0.81$ & $0.83$ & $0.84$ & $0.85$ & $0.86$ & $0.87$ & $1.00$  \\
Campaigns & Seq.
& $0.11$ & $0.15$ & $0.18$ & $0.22$ & $0.26$ & $0.29$ & $0.32$ & $0.35$ & $0.38$ & $0.40$ & $0.43$ & $1.00$ \\
\midrule
Cum. & Dual-optimal
& $0.62$ & $0.75$ & $0.77$ & $0.81$ & $0.83$ & $0.84$ & $0.84$ & $0.85$ & $0.85$ & $0.85$ & $0.86$ & $0.88$ \\
Total & Min.
& $0.42$ & $0.73$ & $0.82$ & $0.86$ & $0.88$ & $0.89$ & $0.89$ & $0.90$ & $0.91$ & $0.91$ & $0.91$ & $0.94$ \\
Value & Seq. 
& $0.19$ & $0.23$ & $0.27$ & $0.30$ & $0.33$ & $0.36$ & $0.38$ & $0.40$ & $0.42$ & $0.44$ & $0.46$ &  $1.44$ \\

\bottomrule
\end{tabular}
\vspace{-0.25cm}
\label{table:result}
\end{table*}
We empirically evaluate the three algorithms discussed in Section~\ref{sec:algorithms}. For confidentiality and advertiser privacy reasons, we evaluate their performance on a semi-synthetic dataset based on actual online advertising auctions. In particular, we focus on advertising campaigns from an online advertising platform that use a bidding product which is captured by our optimization formulation~\eqref{eq:obj_x}. More specifically, an advertiser bids (and therefore also pays) for clicks, i.e., submits bids for cost-per-click,  and the objective is to maximize expected acquisitions (e.g. site visits, calls, conversions) with constraints on total spend being below an input budget and average cost per acquisition below an input target cost ($tcpa$). In our formulation~\eqref{eq:obj_x}, this corresponds to:
$(i)$ The value $v_t$ is equal to $tcpa\cdot pconv_t$, where $pconv_t$ is the probability of a conversion conditioned on a click (note both $tcpa$ and $pconv_t$ are taken to be independent of the bid; while it is obvious for $tcpa$ to be independent of the bid, $pconv$'s independence is supported by empirical studies~\cite{Varian09}); $(ii)$ The allocation $x_t(b_t)$ is the number of clicks won by the advertiser at a bid of $b_t$; $(iii)$  The payment $p_t(b_t)$ is the cost of the clicks won at a bid of $b_t$. 


\subsection{Semi-synthetic Dataset}
Since we study the stochastic setting where the functions $x_t(\cdot),p_t(\cdot)$ are drawn i.i.d. from some distribution, our dataset consists of a set of generative models. The parameters of the generative model for any given (actual) advertising campaign we study are derived from the performance of that campaign in the (actual) auction. 
We discuss the generative model itself and how we pick the parameters for each campaign in more detail in Appendix~\ref{app:empirical-dataset}. 

\subsection{Evaluation Setup}
Our dataset includes $10,000$ randomly selected campaigns, and for each campaign, we set the budget constraint (i.e. $\rho T$ in~\eqref{eq:obj_x}) using its actual daily budget $B$.
We divide the day into $10$-minute periods and use $T=144$. For each campaign, we simulate an algorithm $10$ times to take the average total $\spend$ and total $\conv$ or conversion value as the result of the algorithm on that campaign. We include more details on how we simulate the algorithms as well as visualizations in Appendix~\ref{app:empirical-eval}. For each algorithm, we take the $10,000$ $(\spend,\conv)$ pairs from all the campaigns, and arrange them into buckets based on the relative ROS constraint error\footnote{ROS constraint states that $\spend \leq \conv$. So a constraint violation would imply $\spend > \conv$, i.e., $\spend/\conv-1 > 0$.} $\max\left(0,\spend/\conv - 1\right)$. We look at the cumulative total value achieved by the algorithm through the ROS violation buckets. That is, for the bucket of at most $z\%$ relative error in the ROS constraint, we get the total value over all campaigns such that the algorithm has a relative ROS violation of at most $z\%$. The cumulative total value over ROS error buckets gives us the picture of how an algorithm performs with respect to both the optimization objective and the constraints. 

\paragraph{Benchmark. } For each campaign, our benchmark is the fluid relaxation of~\eqref{eq:obj_x}, but restricted to uniform bidding, i.e., $b_t = k\cdot v_t$ for all $t$. Formally, the benchmark is given by
\begin{equation}
\label{eq:obj_relaxation}
\begin{array}{ll}
\underset{k \ge 0}{\mbox{maximize}} & \sum_{t=1}^T \ex[v_{t}\cdot x_t(k
\cdot v_t)]\\
\mbox{subject to } &  \sum_{t=1}^T \ex[p_t(k\cdot v_t)] \leq  \sum_{t=1}^T \ex[v_{t}x_t(k\cdot v_t)],
\\
& \sum_{t=1}^T \ex[p_t(k \cdot v_t)] \leq \rho T\,.
\end{array}
\end{equation}
We defer the details of how to calculate the benchmark value for campaigns in our semi-synthetic dataset to Appendix~\ref{app:empirical-benchmark}.

\subsection{Results}
We show the performance evaluations of the three algorithms in Table~\ref{table:result}, where each column is associated with a particular error bound, and we show the cumulative fraction of campaigns (top) and cumulative total value of campaigns (bottom) with relative ROS error up to the bound in each column. We normalize the quantities in the table: for value we normalize by our benchmark, i.e. sum of expected opt for all campaigns, and for number of campaigns we normalized by the total number $10000$. We look at the results both in terms of how well the algorithms respect the ROS constraint, and also the optimization objective of value maximization.

\paragraph{ROS constraint.} Both the dual-optimal and min pacing algorithms perform well at keeping the relative ROS error reasonably small, e.g. both have a reasonably large $80\%$ of campaigns finish with at most $20\%$ relative ROS error. 
The sequential pacing algorithm performs poorly in obeying the ROS constraint: only around $11\%$ of campaigns satisfy the ROS constraint, and in Figure~\ref{fig:result}$(a)$ we see a considerable fraction $>20\%$ of campaigns spend more than twice the conversion value.

\paragraph{Value maximization.} Both the dual-optimal and min pacing algorithms also do well at achieving good value. Recall our benchmark on each campaign should be fairly close to the expected optimal value of the fluid relaxation where both the ROS constraint and budget constraint are satisfied on expectation, so it is roughly an upper bound on the expected offline or hindsight optimal, and will be especially meaningful when an algorithm also obey the constraints relatively well. The dual-optimal pacing and min algorithms both achieve very large fraction of the benchmark with fairly small ROS error, e.g., for dual-optimal pacing the campaigns with $\leq 15\%$ relative ROS error in total get $81\%$ of the total benchmark conversion values over all campaigns, and for min pacing it is $86\%$ of the total benchmark. The sequential pacing algorithm gets much smaller total value compared to the dual-optimal and min pacing algorithms over campaigns finishing with small ROS error. 

\paragraph{Stability and convergence} We observe that the trajectory of bidding multipliers generated by the dual-optimal and min pacing algorithms converge to the optimal solution of the benchmark~\eqref{eq:obj_relaxation}. Figure~\ref{fig:example-execution} shows a representative campaign for which the ROS constraint is binding in the benchmark (but the budget constraint is not). After a small learning phase, the dual-optimal pacing algorithm converges to the optimal multiplier of around $k^* \approx 1$. The return-on-spend constraint is mostly obeyed and the total spend is smaller than the budget. 

For the sequential pacing algorithm, however, we do not observe the convergence of bid multipliers. In Figure~\ref{fig:example-execution}, it can be seen that the bid multipliers generated by the sequential pacing algorithm for the same campaign are highly unstable. Moreover, the ROS constraint is violated by a significant amount and the budget is exactly depleted by the end of the horizon.  Interestingly, the behavior of the sequential pacing algorithm is driven by conflicting feedback loops between the budget and ROS pacing services. Recall that, at optimality, only the ROS constraint should bind. Initially, as the ROS pacing service detects a violation of the ROS constraint, it starts increasing its dual variable $\lambda_t$ to satisfy the constraint. This results in a smaller bid multiplier $k_t$ and reduced spend. The budget pacing service, however, believing that the budget constraint is not binding reacts to the lower spend by decreasing its dual variable $\mu_t$, which in turn, results in a higher multiplier. These two opposing feedback loops generate unstable dynamics and one constraint ends up being violated. Similar behaviors are observed across campaigns even when the budget constraint is binding.

\begin{figure}
\centering
\begin{subfigure}{0.49\textwidth}
\centering
\includegraphics[scale=0.3]{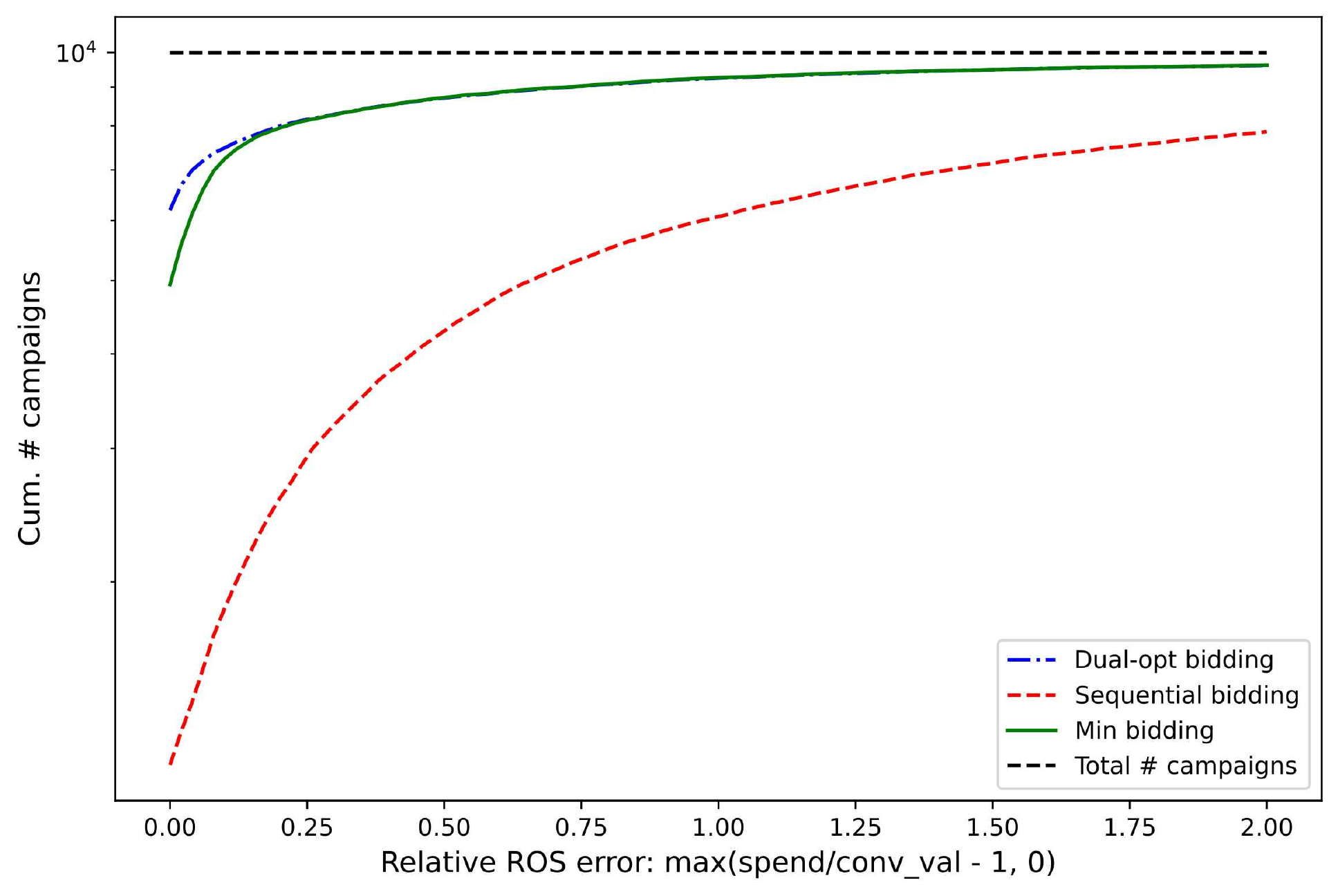}		
\caption{Cumulative number of campaigns}
\label{fig:num_of_campaigns}
\end{subfigure}
\hfill
\begin{subfigure}{0.49\textwidth}
\centering
\includegraphics[scale=0.3]{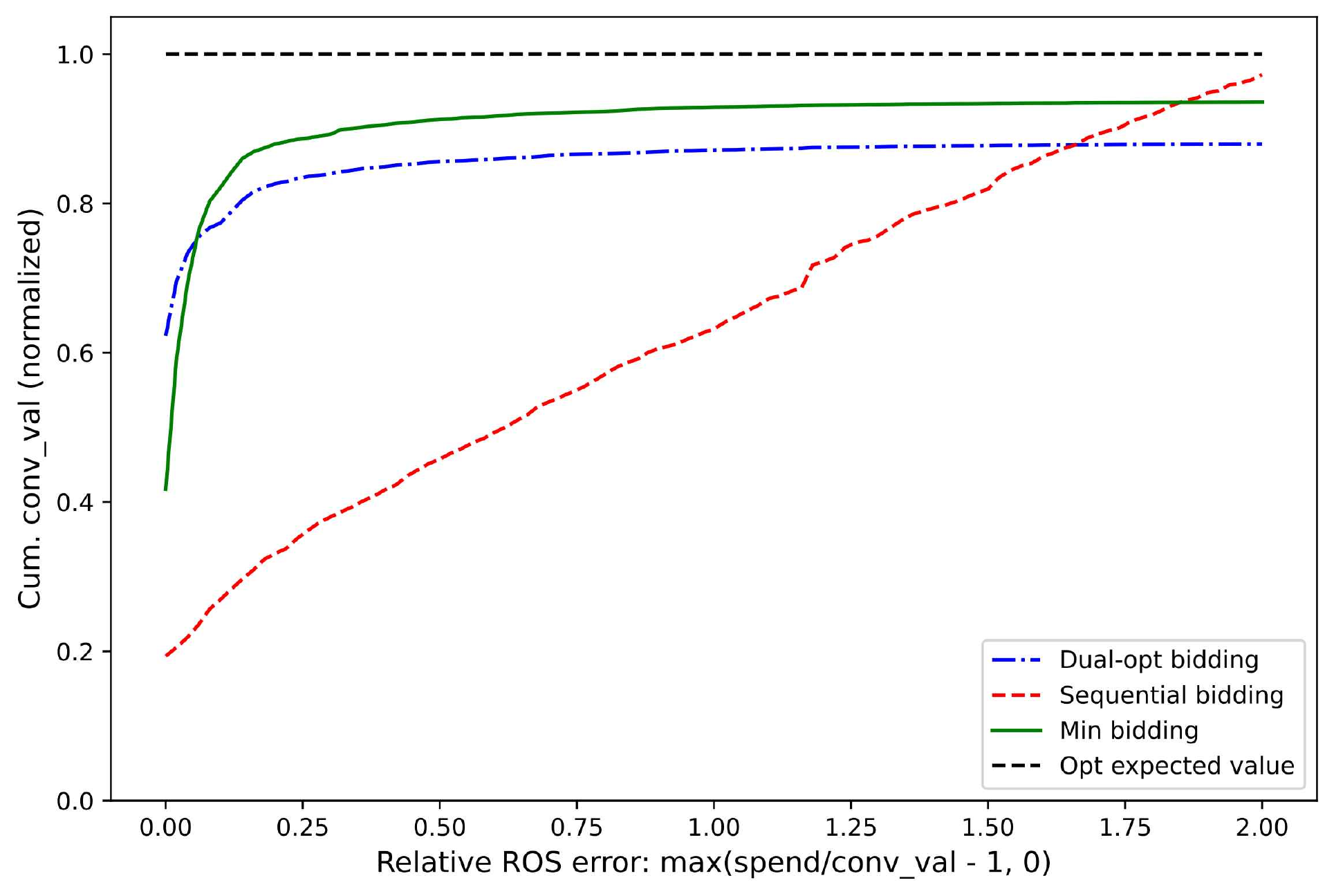}
\caption{Cumulative total \conv}
\label{fig:cumulative_value}
\end{subfigure}
\caption{The cumulative number of campaigns and total \conv{} for each algorithm over the ROS relative error buckets.} 
\label{fig:result}
\end{figure}

\newpage
\bibliographystyle{alpha}
\bibliography{Arxiv_v2/main}

\newpage
\appendix
\section{Related Work}\label{app:related}
Traditional auction theory in microeconomics studies maximizing objectives such as welfare, revenue and gains from trade in the presence of buyer(s) with quasilinear utility, namely, a utility of $v-p$ where $v$ is the value derived and $p$ be the payment. In this work, we adopt a different behavioral model, namely, one where advertisers maximize their value, subject to constraints on the return-on-spend (ROS) and total budget. As mentioned earlier, the significant rise in the adoption of autobidding algorithms in the past few years~\cite{FacebookAutoBidding, GoogleAutoBidding} motivates the study of this model.

\paragraph{Optimal bidding algorithm for a single value-maximizing bidder with budget and/or ROS constraints.} \citet{ABM19} initiated the study of value-maximizing bidders (value maximizers for short) subject to quite general constraints on value and cost. In particular, their model includes budget and ROS constraints. They show how the uniform bidding strategy is optimal if and only if the underlying auction is truthful (where truthfulness is defined from the point-of-view of a quasilinear bidder). Closest to our work is \cite{FPW22} who study the advertiser's value maximization problem in the presence of both budget and ROS constraints in an online repeated auction setting. They show that a specific instantiation of what we call the joint pacing algorithm in this work achieves a $O(\sqrt{T}\log T)$ regret while respecting both the budget and RoS constraints in the stochastic i.i.d. setting. Their algorithm computes the bid as a function of the two Lagrange multipliers exactly as in Equation~\eqref{eq:joint-bid}. 

\paragraph{Welfare in equilibrium among value maximizers.} While the description so far, and also our work, focuses on a single bidder's optimal bidding problem, the equilibrium under the presence of multiple value maximizing bidders has also been a very active area recently.~\citet{ABM19} show how the VCG mechanism, which is welfare maximizing with quasilinear utility maximizers, can achieve, in the worst case, only a fraction $\frac{1}{2}$ of the optimal social welfare. 
Recent work by~\citet{Mehta22} shows how randomization can improve the efficiency beyond the $\frac{1}{2}$ guaranteed by VCG, by establishing a POA of $1.89$ for $2$ bidders and how the POA is unimprovable beyond $2$ even with randomized mechanisms when $n\to\infty$.~\cite{LMP22} study whether non-truthfulness can improve the POA beyond $2$ and show that this is not possible with a deterministic mechanism. But with the combined power of randomization and non-truthful mechanisms, they show how a randomized first-price auction can improve the POA to $1.8$ for two bidders, but again show it is unimprovable beyond $2$ when the number of bidders is large.
Departing from the no information case studied by the above referenced papers, recent works by~\citet{BDMMZ21a,DMMZ21} show how to improve the efficiency under equilibrium beyond $\frac{1}{2}$ by adding boosts and reserves respectively, based on additional information from machine learned advice. 

\paragraph{Revenue-optimal auction for value maximizers with budget and/or ROS constraints.} Much like the design of optimal auctions for utility-maximizing bidders~\citep{Myerson81}, a recent line of work has focused on the design of revenue optimal mechanisms for value maximizers.~\citet{BDMMZ21b,li2020incentive} initiate this line of work, studying the revenue optimal mechanism in the presence of RoS constraints, but no budget constraints, under various information structures regarding whether or not the value is private, whether or not the advertiser specified target is private.~\cite{BDMMZ22} extend this work to include budget constraints for advertisers, and consider the information structure where value is public, so are advertiser budgets, but  advertiser specified target is private.

\paragraph{Optimal bidding algorithm for a single utility maximizing bidder with \& without budget constraint.} While works dealing with budget and ROS constraints in the presence of value maximizers have already been discussed, there has been a long line of work on doing the same for utility maximizers, but usually with just budget constraints. When values and competing bids are drawn from i.i.d. distributions,~\citet{BG19} show that the dual subgradient descent algorithm gives the optimal $O(\sqrt{T})$ regret, and in the adversarial setting they show that it obtains the optimal asymptotic competitive ratio, namely, $B/T$ divided by the maximum value.~\citet{ZCL08} also study pacing in the adversarial setting and give an optimal competitive ratio, but one that is differently parameterized compared to~\cite{BG19}.~\citet{KMS22} study an episodic setting and show how to compute per-period target expenditures based on estimating the probability density based on samples, and ultimately pace based on these target expenditures. On similar lines~\citet{JLZ20} also show how to obtain the optimal $\sqrt{T}$ regret in a non-stationary setting by first learning the probability distributions and then computing target expenditures based on those, using $T\log T$ samples per distribution. Our paper is also loosely related with the rich literature about \emph{Learning to bid in repeated auctions}~\cite{borgs2007dynamics, weed2016online, Feng18, Balseiro19, HZFOW20}, in which the existing papers usually abstract this problem as contextual bandits and do not incorporate budget or ROS constraints into them.

\paragraph{Equilibrium among budget-pacing strategies of utility maximizers.} There is a line of work studying equilibrium outcomes of budget pacing agents interacting with each other. We refer the reader to~\citep{GLLLS22,FT22,CKK21,CKSS22,BBW15} and the references therein for more on this topic. Interestingly, these papers show that uniform bidding is also optimal in the presence of budget constraints. Also,~\citet{BKMM17} perform a comprehensive study of different common budget-pacing strategies and compare the system equilibrium in terms of their welfare, platform revenue, and advertiser utility.

\paragraph{Online resource allocation problems.} The budget pacing problem discussed in the preceding paragraphs is known to be a special case of online resource allocation problems, which have a long line of work. Most of the literature on this topic has focused on the i.i.d.~input model or the slightly more general random permutation model.~\citet{DH09}  introduce a training-based algorithm that learns the optimal dual variables from a batch of initial requests and then uses those to assign the rest of the requests. They show how to obtain a $O(T^{2/3})$ regret for the budgeted allocation problem (also known as the adwords problem) in the random permutation model.~\citet{FHKMS10} obtain a $O(T^{2/3})$ regret for more general linear packing problems in the random permutation model.~\citet{AWY14} obtain an improved $O(\sqrt{T})$ regret by repeatedly solving for the optimal dual variables at geometrically increasing time lengths.~The algorithm of \citet{KTRV14} further solves a linear program at every step and apart from $O(\sqrt{T})$, also obtain the optimal dependence on the number of resources.~\citet{DJSW19} consider more general online packing and covering LPs, but in the i.i.d.~model and obtain a $O(\sqrt{T})$ regret with the optimal dependence on the number of resources. Their algorithm does not need to solve auxiliary linear programs if given an estimate of OPT.~\citep{GM14,AD15,BLM22} make the formal connection between dual descent algorithms and online resource allocation, and show how one can use dual descent algorithms as a black box to obtain a $O(\sqrt{T})$ regret.~In particular, \citep{BLM22,li2020simple} present simple algorithms that do not require solving auxiliary optimization problems.


\pgfmathsetmacro{\picrho}{9/16}
\pgfmathsetmacro{\pickb}{3}
\pgfmathsetmacro{\pickr}{4}
\begin{figure}
    \centering
    \begin{subfigure}{0.4\textwidth}    
        \centering
        \begin{tikzpicture}[scale=1,every node/.style={scale=1}]
            \begin{axis}[xlabel={$k$},ylabel={}, ytick={0}, xtick={0}, yticklabels={}, xticklabels={}, extra x ticks={\pickr},extra x tick labels={$\kr$}, legend pos=north east, legend cell align=left, xmin=0, xmax=8, ymin=-0.25, ymax=0.5,clip=false,axis x line = middle,axis y line = left]       
            \addplot[domain=0:8,color=black,line width=1pt,samples=100]{-(x-4)*x/(2*(x+2)^2)};
            \legend{$\gr(k)$}
            \end{axis}
        \end{tikzpicture}
        \caption{Expected gradient of the ROS constraint $\gr(k) =\E_{v}[ v \cdot x(k \cdot v) - p(k \cdot v)]$.}
        \label{fig:gr}
    \end{subfigure}
    \hspace{1cm}
    \begin{subfigure}{0.4\textwidth}    
        \centering
        \begin{tikzpicture}[scale=1,every node/.style={scale=1}]
            \begin{axis}[xlabel={$k$},ylabel={}, ytick={0}, xtick={0}, yticklabels={}, xticklabels={}, extra x ticks={\pickb},extra x tick labels={$\kb$}, legend pos=north east, legend cell align=left, xmin=0, xmax=8, ymin=-0.5, ymax=1,clip=false,axis x line = middle,axis y line = left]       
            \addplot[domain=0:8,color=black,line width=1pt,samples=100]{\picrho- x*x/(x*x+2*x+1)};
            \legend{$\gb(k)$}
            \end{axis}
        \end{tikzpicture}
        \caption{Expected gradient of the budget constraint $\gb(k)=\rho - \E_{v} [p(k \cdot v)]$.}
        \label{fig:gb}        
    \end{subfigure}
    \caption{Expected gradients for an example in which values and competing bids are independent and exponentially distributed with means 1/2 and 1, respectively. Also, $\rho = 9/16$. Both curves cross the positive $k$-axis once and from above (Assumption~\ref{ass:single-crossing}). Also, strong monotonicity holds for this example (Assumption~\ref{ass:strongly-monotone}).}
    \label{fig:gradients}
\end{figure}
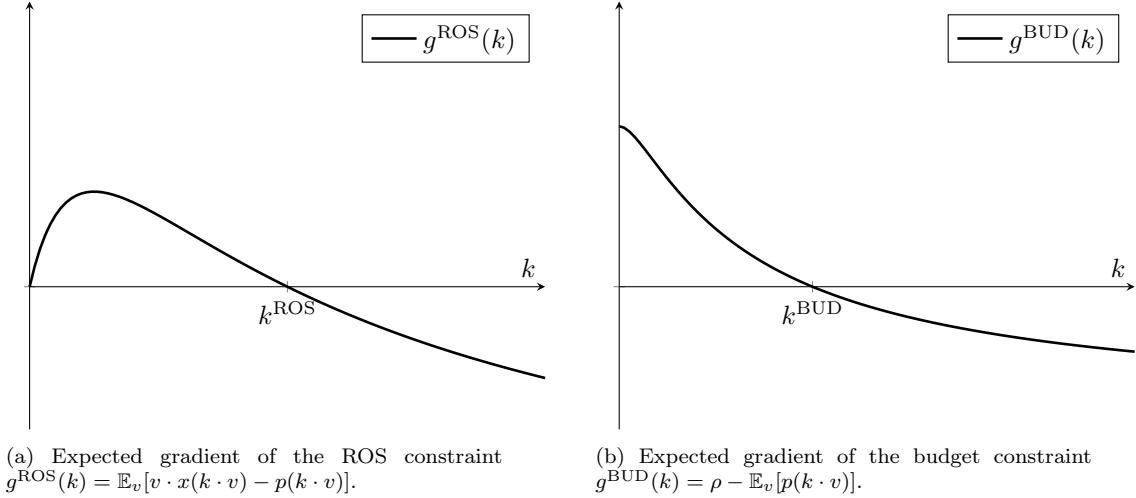
\section{Proof of the regret bound of the $\mpacing$ algorithm}\label{app:min_proof}
We choose the orbit to be a ball of size $\epsilon>0$ around the optimal solution
\[
    \mathcal O_\epsilon = \left\{ \{(\lambda,\mu) \in \mathbb R_+^2 : \max\left( |\lambda - \lambda^*|, |\mu - \mu^*|\right) < \epsilon\right\}\,.
\]
The value of $\epsilon>0$ is chosen so that 
\begin{enumerate}
    \item Assumption~\ref{ass:strongly-monotone} is satisfied for all $(\lambda,\mu) \in \mathcal O_\epsilon$,
    \item The algorithm bids according to the binding constraint for all $(\lambda,\mu) \in \mathcal O_\epsilon$, i.e., $k^{\min}(\lambda,\mu) = (1+\lambda)/\lambda$, 
    \item The gradient of the budget constraint satisfies $\gb(k^{\min}(\lambda,\mu)) > 0$ for all $(\lambda,\mu) \in \mathcal O_\epsilon$,
    \item The gradients satisfy $\gr(k^{\min}(\lambda,\mu)) < 0$ and $\gb(k^{\min}(\lambda,\mu)) < 0$ if either $\lambda < \epsilon$ or $\mu < \epsilon$.
\end{enumerate}
The second condition can be satisfied by Lemma~\ref{lemma:optimal-multipler} because there exists an optimal dual optimal solution with  $\mu^*=0$ and $\lambda^* = 1/(\kr - 1) > 0$, and the algorithm bids according to the ROS multiplier when $(1+\lambda)/\lambda < 1/\mu$. The third condition can be satisfied because the single-crossing property (Assumption~\ref{ass:single-crossing}) implies that $\gb(k) > 0$ for $k < \kb$ and non-degeneracy (Assumption~\ref{ass:nondegenerate}) implies that for $\kr < \kb$. The fourth condition holds by the single-crossing property because the gradients are negative for large enough multipliers $k$.  

\pgfmathsetmacro{\utau}{1}
\pgfmathsetmacro{\aa}{3}
\pgfmathsetmacro{\bb}{1/10}

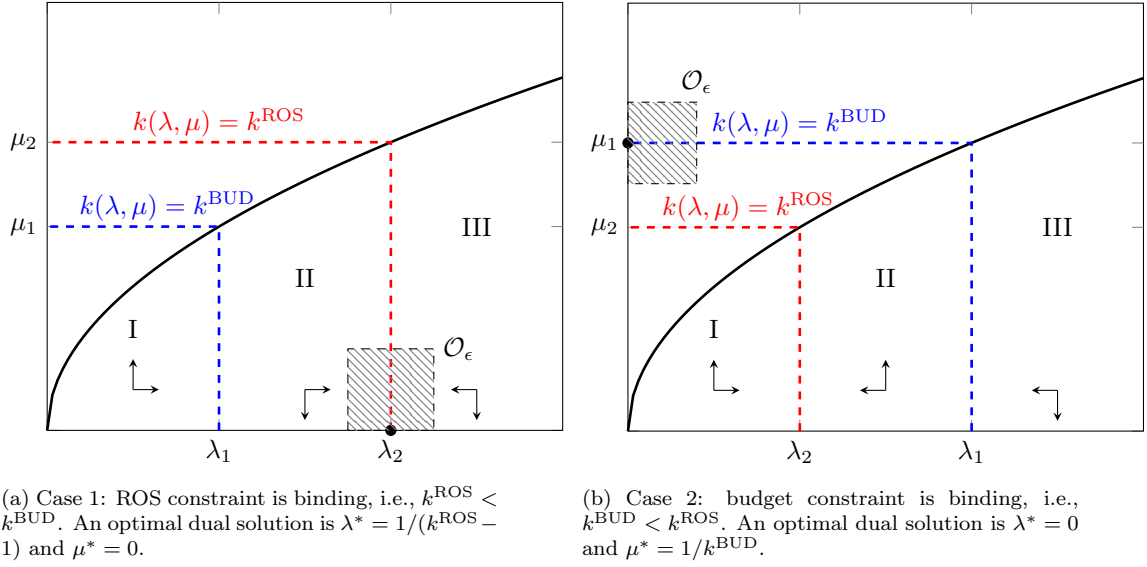
\begin{figure}
\centering
\begin{subfigure}{0.4\textwidth}
  \centering
  \begin{tikzpicture}[scale=1,every node/.style={scale=1}]
        \begin{axis}[xlabel={},ylabel={}, ytick={0}, xtick={0}, yticklabels={}, xticklabels={}, extra x ticks={\utau, 2},extra x tick labels={$\lambda_1$, $\lambda_2$},  extra y ticks={\utau, 1.414}, extra y tick labels={$\mu_1$, $\mu_2$},legend pos=outer north east, legend cell align=left, xmin=0, xmax=3, ymin=0, ymax=2.1,clip=false]
        \addplot[domain=0:3,color=black,line width=1pt,samples=100]{ sqrt(x)};
        \addplot[forget plot,color=blue,line width=1pt,dashed] coordinates {(\utau,0) (\utau,\utau) (0,\utau)};
        \addplot[forget plot,color=red,line width=1pt,dashed] coordinates {(2,0) (2,1.414) (0,1.414)};
        \node at (axis cs: 0.5,0.5) (R1) {I};
        \draw [-stealth](axis cs: 0.50,0.20) -- (axis cs: 0.50,0.35);
        \draw [-stealth](axis cs: 0.50,0.20) -- (axis cs: 0.65,0.20);
        \node at (axis cs: 1.5,0.75) {II};
        \draw [-stealth](axis cs: 1.50,0.20) -- (axis cs: 1.50,0.05);
        \draw [-stealth](axis cs: 1.50,0.20) -- (axis cs: 1.65,0.20);
        \node at (axis cs: 2.5,1) {III};
        \draw [-stealth](axis cs: 2.50,0.20) -- (axis cs: 2.50,0.05);
        \draw [-stealth](axis cs: 2.50,0.20) -- (axis cs: 2.35,0.20);
        \node at (axis cs: 0.7,1.1) {\textcolor{blue}{{$k(\lambda,\mu)=\kb$}}};
        \node at (axis cs: 1,1.51) {\textcolor{red}{{$k(\lambda,\mu)=\kr$}}};
        \node at (axis cs: 2,0)[circle,fill,inner sep=1.5pt]{};
        \draw[pattern=north west lines, pattern color=gray,dashed] (axis cs: 1.75,0) rectangle (axis cs:2.25,0.4) node[right] {$\mathcal O_\epsilon$} ;        
        \end{axis}
        \end{tikzpicture}
  \caption{Case 1: ROS constraint is binding, i.e., $\kr < \kb$. An optimal dual solution is $\lambda^*=1/(\kr - 1)$ and $\mu^*=0$.}
  \label{fig:case1}
\end{subfigure}%
\hspace{1cm}
\begin{subfigure}{0.4\textwidth}
    \centering
    \begin{tikzpicture}[scale=1,every node/.style={scale=1}]
        \begin{axis}[xlabel={},ylabel={}, ytick={0}, xtick={0}, yticklabels={}, xticklabels={}, extra x ticks={\utau, 2},extra x tick labels={$\lambda_2$, $\lambda_1$},  extra y ticks={\utau, 1.414}, extra y tick labels={$\mu_2$, $\mu_1$},legend pos=outer north east, legend cell align=left, xmin=0, xmax=3, ymin=0, ymax=2.1,clip=false]
        \addplot[domain=0:3,color=black,line width=1pt,samples=100]{ sqrt(x)};
        \addplot[forget plot,color=blue,line width=1pt,dashed] coordinates {(2,0) (2,1.414) (0,1.414)};
        \addplot[forget plot,color=red,line width=1pt,dashed] coordinates {(\utau,0) (\utau,\utau) (0,\utau)};
        \node at (axis cs: 0.5,0.5) {I};
        \draw [-stealth](axis cs: 0.50,0.20) -- (axis cs: 0.50,0.35);
        \draw [-stealth](axis cs: 0.50,0.20) -- (axis cs: 0.65,0.20);        
        \node at (axis cs: 1.5,0.75) {II};
        \draw [-stealth](axis cs: 1.50,0.20) -- (axis cs: 1.50,0.35);
        \draw [-stealth](axis cs: 1.50,0.20) -- (axis cs: 1.35,0.20);
        \node at (axis cs: 2.5,1) {III};
        \draw [-stealth](axis cs: 2.50,0.20) -- (axis cs: 2.50,0.05);
        \draw [-stealth](axis cs: 2.50,0.20) -- (axis cs: 2.35,0.20);        
        \node at (axis cs: 1,1.51) {\textcolor{blue}{{$k(\lambda,\mu)=\kb$}}};
        \node at (axis cs: 0.7,1.1) {\textcolor{red}{{$k(\lambda,\mu)=\kr$}}};
        \node at (axis cs: 0,1.414)[circle,fill,inner sep=1.5pt]{};
        \draw[pattern=north west lines, pattern color=gray,dashed] (axis cs: 0,1.414-0.2) rectangle (axis cs:0.4,1.414+0.2) node[above] {$\mathcal O_\epsilon$} ;        
        \end{axis}
    \end{tikzpicture}
    \caption{Case 2: budget constraint is binding, i.e., $\kb < \kr$. An optimal dual solution is $\lambda^*=0$ and $\mu^*=1/\kb$.}
    \label{fig:case2}
\end{subfigure}
\caption{Illustration of the two cases for the MIN dynamics. The black dot indicates an optimal solution and the hatched rectangle is an orbit $\mathcal O_\epsilon$ of size $\epsilon$ around the optimal solution. The solid black curve gives the points for which the $(1+\lambda)/\lambda = 1/\mu$, i.e., the multipliers of both constraints are equal. Above the curve, the algorithm bids $1/\mu$ according to the budget constraint, and below it bids $(1+\lambda)/\lambda$ according to the ROS constraint. The arrows indicate the drift of the stochastic process in each region. The red (blue, resp.) dashed curve gives the set of dual variables for which $k^{\min}(\lambda,\mu)=\kr$ ($=\kb$, resp.).}
\end{figure}

\subsubsection{Step 1: Binding Constraint Identification}

For the first step, we show that if dual variables are positive, it takes the ODE a constant amount of time to get to the interior of the orbit.

\begin{lemma}\label{lemma:ode-path} For any initial dual solution $(\lambda(0),\mu(0)) \not\in [0,\epsilon)^2$, there exits a finite time $\sigma > 0$ such that the solution of \eqref{eq:ode} satisfies $(\lambda(\sigma), \mu(\sigma)) \in \mathcal O_{\epsilon/2}$ and for all $s \in [0,\sigma]$ we have $(\lambda(s),\mu(s)) \not\in [0,\epsilon)^2$.
\end{lemma}

\begin{proof}
It follows from Assumption \ref{ass:single-crossing} and \ref{ass:nondegenerate} that $\gr(k) > 0$ for $k< \kr$ and $\gr(k)<0$ for $k>\kr$. Furthermore, denote $\mu_1=1/\kb$, $\mu_2=1/\kr$, $\lambda_1=1/(\kb-1)$, $\lambda_2=1/(\kr-1)$. We consider two cases depending on whether $\kb$ or $\kr$ is smaller.

\textbf{Case 1: $\kb>\kr$.}
In this case, we have $k^*=\kr$, and the unique stationary point is given by $\mu^*=0$ and $\lambda^*=\lambda_2 = 1/(\kr-1)$. 

The whole space can be split into three regions (see Figure \ref{fig:case1}):

\textbf{Region I: $k>\kb$.} This region corresponds to $\{(\mu,\lambda):\mu<\mu_1, \lambda<\lambda_1\}$. In this region, we have $\gr(k)<0$ and $\gb(k)<0$, thus $\dmu>0$ and $\dlambda>0$. 

\textbf{Region II: $\kr<k<\kb$.} This region corresponds to $\{(\mu,\lambda):\mu<\mu_2, \lambda<\lambda_2\}$ subtracting region I. In this region, we have $\gb(k)>0$ and $\gr(k)<0$, thus $\dmu<0$ and $\dlambda>0$.

\textbf{Region III: $k>\kr$.} This region corresponds to the complementary set of region I and II. In this region, we have $\gb(k)>0$ and $\gr(k)>0$, thus $\dmu<0$ and $\dlambda<0$.

Now, we are ready to show the result. Before proceeding, note that by definition $\epsilon$, we have that $\dot \mu > 0$  and $\dot \lambda > 0$ if $\mu < \epsilon$ and $\lambda < 0$. Therefore, the ODE can never get closer to a distance $\epsilon$ from the origin.

First, we claim that for any initial solution $\mu(0), \lambda(0)$, there exists $s_1$ such that it holds for all $s>s_1$ that $\mu(s)\le \hmu:= \frac{1}{2}\pran{\mu_1+\mu_2}$. This is because once $\mu(s)\le \hmu$, $\mu(s)$ would never go above $\hmu$ due to the dynamics in regions II and III. So we just need to consider the first time $\mu(s)\le \hmu$. Notice that for all $(\mu,\lambda)$ such that $\mu> \hat \mu$, there exists $\delta_1$ such that we have $\dmu < \delta_1<0 $. Thus, we just need to choose $s_1= \frac{1}{|\delta_1|}((\mu(0)-\mu_1)^+)$.

Second, we claim there exists $s_2>s_1$ such that for $s>s_2$, we have $\mu(s)\le \hmu$ and $\lambda(s) \ge \hlambda:= \frac{1}{2}(\lambda_1+\lambda_2)$. This is because after $s_1$, $\mu(s)\le \hmu$. Thus, once $\lambda(s) \ge \hlambda$, $\lambda(s)$ would never go below $\hlambda$ due to the dynamics in the regions I and II. So we just need to consider the first time $\lambda(s) \ge \hlambda$. Notice that for all $(\mu,\lambda)$ such that $\mu\le \mu_1, \lambda\le \hlambda$, there exists $\delta_2$ such that we have $\dlambda \ge \delta_2>0 $. Thus, we just need to choose $s_2=s_1+ \frac{1}{\delta_2}((\hlambda-\lambda(s_1))^+)$.

Third, we claim there exists $s_3>s_2$ such that for $s>s_3$, we have $\mu(s) \le \epsilon/2$ and $\lambda(s) \ge \hlambda$. This is because after $s_2$, $\mu(s)\le \hmu, \lambda(s) \ge \hlambda$. 
In this region, there exists $\delta_3 < 0$ such that $\dmu\le \delta_3 \mu <0$ and we just need to choose $s_3=s_2+ \log(\epsilon/(2\mu))/|\delta_3|$.

Fourth, we claim there exists $s_4 > s_3$ such that for $s > s_4$, we have that $\mu(s) \le \epsilon/2$ and $|\lambda(s) - \lambda_2| \le \epsilon/2$. This is because after $s_3$ we have that $\dot \mu \le 0$ and hence $\mu(s) \le \epsilon/2$ for all $s>s_3$. The single-crossing property implies that $\dot \lambda = 0$ only at $\lambda_2$, so we should reach $|\lambda(s) - \lambda_2| \le \epsilon/2$ in finite time.

\textbf{Case 2: $\kb<\kr$.}
This case is exactly symmetric to Case 1 by flipping $\mu$ and $\lambda$ (see Figure \ref{fig:case2}).
\end{proof}

We invoke the following result, which bounds the maximum error between a discrete-time stochastic system and its continuous-time ODE approximation.

\begin{lemma}\label{lemma:ode-deviations} Consider the stochastic process $\{Y_t\}_{t\ge0}$ with $Y_t \in \mathbb R^n_{++}$ satisfying
\[
    Y_{t+1} = Y_t + \alpha h_t(Y_t)\,,
\]
where $h_t :\mathbb R^n \mapsto \mathbb R^n$ is a random function and $\alpha > 0$ is the step-size. The initial state $Y_0$ lies in an open subset $\mathcal Y \subseteq \mathbb R^n$. The random functions are i.i.d.~with expectation $\E h_t(y) = \bar h(y)$. We assume that the random functions have uniformly bounded expectation $\bar h_i(y) \le \bar H$ for all $y \in \mathcal Y$, uniformly bounded variance $\operatorname{Var}[h_{t,i}(y)] \le \bar H_2$ for all $y \in \mathcal Y$, and its expectation is $L$-Lipschitz continuous in $\mathcal Y$ w.r.t.~the max-norm, i.e., $\|\bar h(y) - \bar h(y')\|_\infty \le L \|y - y'\|_\infty$ for all $y,y'\in \mathcal Y$. Then, the following holds:
\begin{enumerate}
    \item The ODE $\frac d {ds} \bar Y(s) = \bar h(\bar Y(s))$ with $\bar Y(0) = Y_0 \in \mathcal Y$ has a unique solution in $\mathcal Y$.

    \item Fix $\epsilon > 0$. Let $\sigma \ge 0$ be such that $\| \bar Y(s) - y\|_\infty > \epsilon$ for all $s \in [0,\sigma]$ and $y \not\in \mathcal Y$. Then, 
    \[
    \mathbb P\left\{ \max_{t : \alpha t \le \sigma \ } \left\| Y_t - \bar Y(\alpha t) \right\| _\infty > \epsilon \right\} \le \epsilon^{-2}
    \left( \alpha \sigma L \bar H + \sqrt{4 n \alpha \sigma \bar H_2} \right)^2 \exp (2 L \sigma)
    \]
\end{enumerate}
\end{lemma}

\begin{proof}
Denote by $s_t = \alpha t$ the corresponding time in the ODE for the discrete step $t$. The first part follows from Picard–Lindel\"of theorem because $\bar h$ is Lipschitz continuous.

We prove the second part in two steps. In the first step, use the Lipschitz continuity of the dynamics to show that deviations of $Y_t$ from the expected path $\bar Y(s_t)$ accumulate linearly and conclude by using a discrete version of Gronwall's Lemma to bound the absolute deviations in an almost sure sense. This first step performs a deterministic analysis of the deviations. In the second step, we use a concentration argument to bound the maximum deviation in a stochastic sense.

\paragraph{Step 1.} Introduce a time $\tau$ corresponding to the first time $t$ with $s_{t+1} \le \sigma$ such that $Y_{t+1} \not \in \mathcal Y$. Consider a step $t \ge 1$ under the event that $t \le \tau$, which implies that $Y_j \in \mathcal Y$ for all $j \le t$ and $\bar Y(s) \in \mathcal Y$ for all $s \le s_t$. Using the dynamics of the stochastic process and the ODE, we obtain that
\begin{align*}
    Y_{t+1} - \bar Y(s_{t+1}) &= Y_t - \bar Y(s(t)) + \alpha h_t ( Y_t) - \int_{s_t}^{s_{t+1}} \bar h (\bar Y(s)) ds\,.
\end{align*}
From the mean value theorem, because the solution to the ODE is absolutely continuous, we know there exists $\zeta_i \in [s_t, s_{t+1}]$ such that
\begin{align*}
    \int_{s_t}^{s_{t+1}} \bar h_i (\bar Y(s)) ds &= (s_{t+1} - s_t) \bar h_i (\bar Y(\zeta_i))\\
    &= \alpha \bar h_i (Y_t) + \underbrace{\alpha \left( \bar h_i (\bar Y(s_t)) - \bar h_i(Y_t) \right) + \alpha \left(  \bar h_i(\bar Y(\zeta_i)) - \bar h_i (\bar Y(s_t)) \right)}_{\beta_{t,i}}\,.
\end{align*}
Therefore, we have that
\begin{align*}
    Y_{t+1} - \bar Y(s_{t+1}) &= Y_t - \bar Y(s(t)) + \alpha \Delta_t + \beta_t\,.
\end{align*}
where $\Delta_t = h_t(Y_t) - \bar h(Y_t)$. We refer to $\Delta_t$ as a stochastic error and $\beta_t$ as the integration error. Using that $h$ is $L$-Lipschitz continuous in $\mathcal Y$, the integration error can be bounded as follows:
\begin{align*}
    |\beta_{t,i}| &\le \alpha   \left| \bar h_i (\bar Y(s_t)) - \bar h_i(Y_t) \right| + \alpha \left|(  \bar h_i(\bar Y(\zeta_i)) - \bar h_i (\bar Y(s_t)) \right|\\
    &\le \alpha L \|Y_t - \bar Y(s_t)\|_\infty + \alpha L \|\bar Y(\zeta_i) - \bar Y(s_t)\|_\infty\\
    &\le  \alpha L \|Y_t - \bar Y(s_t)\|_\infty + \alpha^2 L \bar H\,,
\end{align*}
where the last inequality follows because from the mean value theorem there exists $\zeta_j'' \in [s_t, \zeta_i]$ such that $| \bar Y_j(\zeta_i) - \bar Y_j(s_t)| = |(\zeta_i - s_t) \bar h( Y(\zeta_j''))| \le 
\alpha \bar H$  together with the fact that $| \zeta_i - s_t| \le \alpha$ and $|\bar h(y)| \le \bar H$.

Therefore, summing over steps $j = 0,\ldots,t$ and using that the initial conditions satisfy $Y_0 = \bar Y(0)$, we obtain that the following is true under the event $t\le \tau$:
\begin{align*}
    \left\|  Y_{t+1} - \bar Y(s_{t+1}) \right\|_\infty 
    &= \Bigg\|  \sum_{j=0}^{t} \left( \alpha \Delta_j + \beta_j \right) \Bigg\|_\infty\\
    &\le \alpha \left\| M_{t} \right\|_\infty + \alpha L \sum_{j=1}^{t} \|Y_j - \bar Y(s_j)\|_\infty + \alpha s_{t+1} L \bar H\,,
\end{align*}
where the we denote by $M_t = \sum_{j=0}^{t} \Delta_j = \sum_{j=0}^t h_j(Y_j) - \bar h(Y_j)$ and last inequality follows from the triangle inequality together with $s_t = \alpha t$.

We  next apply the following discrete version of Gronwall's Lemma.

\begin{lemma}[Discrete Gronwall's Lemma] Let $x_t \ge 0$ be a sequence satisfying $x_t \le a+b \sum_{j=1}^{t-1} x_j$ with $a,b\ge 0$. Then, $x_t \le a \exp(bt).$
\end{lemma}

Setting $x_t = \left\|  Y_t - \bar Y(s_t) \right\|_\infty$ and choosing $a,b$ appropriately, we obtain that
\[
     \left\|  Y_{t+1} - \bar Y(s_{t+1}) \right\|_\infty \le 
     \left( \alpha s_{t+1} L \bar H + \alpha \max_{\ell=0,\ldots,t} \left\| M_\ell \right\|_\infty \right) \exp (L s_{t+1})\,.
\]

\paragraph{Step 2.} Denote by $\mathcal F_t = \sigma(h_0, \ldots, h_t)$ the sigma-algebra generated by the random functions up to step $t$. We have that $M_t$ is a martingale because $M_t \in \mathcal F_t$ and $\E[M_{t+1} | \mathcal F_t] = M_t$. Moreover, $Y_0,\ldots, Y_{t+1} \in \mathcal F_t$ and $\tau$ is a stopping time with respect to $\mathcal F_t$ because $\tau \in \mathcal F_t$.

Taking expectations over the maximum of all steps up to $\tau$, we obtain that
\begin{align*}
    \left( \mathbb E\left[ \max_{j=1,\ldots,\min(t+1,\tau)}  \left\|  Y_j - \bar Y(s_j) \right\|_\infty^2\right] \right)^{1/2} 
    \le \left( \alpha s_{t+1} L \bar H + \alpha \left( \mathbb E\left[\max_{j=0,\ldots,\min(t,\tau)} \left\| M_j \right\|_\infty^2\right] \right)^{1/2}  \right) \exp (L s_{t+1})\,,
\end{align*}
where the first inequality follows from Minkowski inequality. It is sufficient to bound each coordinate at a time because
\[
    \max_{j=0,\ldots,\min(t,\tau)} \left\| M_j \right\|_\infty^2
    =\max_{i=1,\ldots,n} \max_{j=0,\ldots,\min(t,\tau)} \left|M_j \right|^2 \le \sum_{i=1}^n \max_{j=0,\ldots,\min(t,\tau)} \left| M_{j,i} \right|^2\,,
\]
where the first equation follows from exchanging maximums and the second since $\|x\|_\infty \le \|x\|_1$. Using that $\tau$ is a stopping time and $M_t$ is a martingale that
\begin{align*}
    \E\left[\max_{j=0,\ldots,\min(t,\tau)} \left| M_{j,i} \right|^2\right] &= 
    \E\left[\max_{j=0,\ldots,t} \left| M_{\min(j ,\tau),i} \right|^2\right]\\ 
    & \le 4\E\left[M_{\min(t ,\tau),i}^2\right]\\
    &= 4\E\left[\left(\sum\nolimits_{j=0}^t \Delta_j \mathbf 1\{j \le t\} \right)^2 \right]
    = 4\sum_{j=0}^t \E\left[\Delta_j^2 \mathbf 1\{j \le t\} \right]\\
    &\le 4\sum_{j=0}^t \E\left[\Delta_j^2  \right]
    \le 4 (t+1) \bar H_2\,,
\end{align*}
where the first inequality follows from Doob's Martingale Inequality because the stopped martingale $M_{\min(t ,\tau),i}$ is a martingale,
the second equality because martingale differences are orthogonal, and the last our bound on the variance of the random function. Putting everything together, we obtain that
\begin{align}\label{eq:l2-norm-bound}
    \left( \mathbb E\left[ \max_{j=0,\ldots,\min(t+1,\tau)}  \left\|  Y_t - \bar Y(s_t) \right\|_\infty^2\right] \right)^{1/2} 
    \le \left( \alpha s_{t+1} L \bar H + \sqrt{4 n \alpha s_{t+1} \bar H_2} \right) \exp (L s_{t+1})\,.
\end{align}
To conclude that if $t$ is the first time win which $\|Y_t - \bar Y(s_t)\|_\infty > \epsilon$, then we must have $\|Y_{t-1} - \bar Y(s_{t-1})\|_\infty \le \epsilon$, which implies that $Y_{t-1} \in \mathcal Y$ (because if $Y_{t-1} \not\in \mathcal Y$, we would have that $\|Y_{t-1} - \bar Y(s_{t-1})\|_\infty$ because $\alpha {t-1} \le \sigma$ and the definition of $\sigma$). The latter imples that $\tau \ge t-1$ or $t+1 \le \tau$. Therefore, we can write the event in the statement as
\begin{align*}
     \mathbb P\left\{ \max_{t : \alpha t \le \sigma } \left\| Y_t - \bar Y(\alpha t) \right\|_\infty \ge \epsilon \right\} 
     &= \mathbb P\left\{ \max_{t : \alpha t \le \sigma } \left\| Y_t - \bar Y(\alpha t) \right\| _\infty \mathbf 1\{ t+1 \le \tau \} \ge \epsilon \right\}\\
     &= \mathbb P\left\{ \max_{t : \alpha t \le \sigma, t \le \tau-1} \left\| Y_t - \bar Y(\alpha t) \right\| _\infty  \ge \epsilon \right\}\\
     &\le \epsilon^{-2} \E\left[\max_{t : \alpha t \le \sigma, t \le \tau-1} \left\| Y_t - \bar Y(\alpha t) \right\| _\infty^2 \right]\\     
     &\le \epsilon^{-2} \left( \alpha s_{t} L \bar H + \sqrt{4 n \alpha s_{t} \bar H_2} \right)^2 \exp (2 L s_{t})\,,
\end{align*}
where the first inequality follows from an application of Markov's inequality and the last from \eqref{eq:l2-norm-bound}. We conclude by noting that $s_t \le \sigma$.
\end{proof}

We apply Lemma~\ref{lemma:ode-deviations} to $Y_t = (\log \lambda_t, \log \mu_t)$ and set the random function $h_t$ to be the gradients of the ROS and budget constraints, respectively. That is,
\[
    h_t(y) = h_t(\log \lambda, \log \mu) = -\left( v_t \cdot x_t\left( v_t \cdot k^{\min}(\lambda, \mu)\right) - p_t\left( v_t \cdot k^{\min}(\lambda, \mu)\right), \rho - p_t\left( v_t \cdot k^{\min}(\lambda, \mu)\right) \right)\,.
\]
This choice reduces the stochastic process in the statement of the lemma to the update rule of the algorithm. Taking expectations, we obtain that
\[
    \bar h(\log \lambda, \log \mu) = -\left( \gr \left( k^{\min}(\lambda, \mu)\right), \gb\left( k^{\min}(\lambda, \mu)\right) \right)
\]
By assumption, the expected gradients are bounded and have finite variance. For Lipschitz continuity we need to show that for $g = \gb,\gr$
\[\left|g(k^{\min}(\lambda, \mu)) - g(k^{\min}(\lambda', \mu')) \right| \le L \max\left(\|\log \lambda - \log \lambda'\|, \|\log \mu - \log \mu'\|\right).
\]
The expected gradients, however, are not Lipschitz continuous for all multipliers because of the logarithmic transformation. To guarantee Lipschitz continuity, we restrict the set of dual solutions to lie in the set 
\[
    \mathcal Y = \left\{ (\log(\lambda), \log(\mu)) \in \mathbb R^2 : \lambda > \epsilon/2 \text{ or } \mu > \epsilon/2\right\}\,.
\]
For example, the gradient of the budget constraint be written as
\begin{align*}
    \gb\left( k^{\min}(\lambda, \mu)\right)
    &= \gb\left( k^{\min}(\exp(\log(\lambda), \exp(\log(\mu))\right)\\
    &= \gb\left( \min\left(\exp(-\log(\lambda)) + 1\right), \exp(-\log(\mu)) \right)\\
    &= \gb\left( \exp \min\left(\log( \exp(-\log(\lambda)) + 1), -\log(\mu)\right) \right)\,.
\end{align*}
Because the minumum $\min(x,y)$ and the log-sum-exp function $\log(\exp(x) + 1)$ are 1-Lipschitz continuous, we obtain that 
\[\min\left(\log( \exp(-\log(\lambda)) + 1), -\log(\mu)\right)
\]
is 1-Lipschitz continuous in $(\log \lambda, \log \mu)$. For $(\log \lambda, \log \mu) \in \mathcal Y$ we have that $k^{\min}(\lambda, \mu) < 2/\epsilon$, which implies that $\gb\left( k^{\min}(\lambda, \mu)\right)$ is $2L_g/\epsilon$-Lipscthiz continuous because the exponential $\exp(x)$ function is $\exp(a)$-Lipschitz continuous in $[0,a]$.

In Lemma~\ref{lemma:ode-deviations}, we set $\sigma$ as the time it takes the ODE to reach the set $\mathcal O_{\epsilon/2}$ and $\epsilon := \epsilon/2$. Under the good event $A = \left\{\max_{t : \alpha t \le \sigma \ } \left\| Y_t - \bar Y(\alpha t) \right\| _\infty \le \epsilon/2\right\}$, we have by Lemma~\ref{lemma:ode-path} that $(\lambda_t, \mu_t) \not\in[0,\epsilon]^2$ and, thus, the dynamics are Lipschitz continuous. Moreover, we because the step size is $\alpha \approx T^{-1/2}$ we have that at time $\tau = \lfloor \sigma / \alpha \rfloor = O\left(T^{1/2}\right)$ the state of the algorithm reaches the orbit $\mathcal O_\epsilon$ with high probability. More formally, we have proved the following result.

\begin{prop}\label{prop:prob-identification} For every initial dual solution $(\lambda_1,\mu_1) \not \in[0,\epsilon)^2$, there exists a time $\tau = O\left(T^{1/2}\right)$ such that the probability of not hitting the orbit is bounded by
\[
    \mathbb P\left\{ (\lambda_\tau, \mu_\tau) \not\in \mathcal O_\epsilon \right\} = O\left(T^{-1/2}\right)\,.
\]    
\end{prop}

\subsubsection{Step 2: Orbital Stability}

We next show that once the iterates reach the orbit $\mathcal O_\epsilon$, they stay in the orbit for the rest of the horizon with high probability. To prove this result we show that the sum of Bregman divergence $V_h$ induced by the negative entropy $h(u) = u \log u$ constitutes a stochastic Lyaponuv function. The Lyaponuv function is given by
\[
    V(\lambda,\mu) = V_h(\lambda^*, \lambda) + V_h(\mu^*, \mu)\,,
\]
where the Bregman divergence $V_h(y, x) = h(y)-h(x) - h^\prime(x)\cdot (y-x)$ is $V_h(y, x) = y \log (y/x) - y + x$. Note that we can choose $m > 0$ such that $V(\lambda,\mu) < m$ for all $(\lambda,\mu) \in \mathcal O_\epsilon$.

Assume that the ROS constraint is binding so that $\mu^* = 0$. Here, we have that $V_h(\mu^*, \mu) = \mu$. Let 
\[
\gr_t = v_t \cdot x_t\left( v_t \cdot k^{\min}(\lambda_t, \mu_t)\right) - p_t\left( v_t \cdot k^{\min}(\lambda_t, \mu_t)\right) \quad \text{and} \quad \gb_t = \rho - p_t\left( v_t \cdot k^{\min}(\lambda_t, \mu_t)\right)
\]
be the empirical gradients at time $t$. The multiplicative weight update implies that
\[
    V(\lambda_{t+1},\mu_{t+1}) \le 
    V(\lambda_{t},\mu_{t})
    - \alpha \gr_t \cdot (\lambda_t - \lambda^*) 
    - \alpha \gb_t \cdot  (\mu_t - \mu^*)
    + \alpha^2 \frac{\bar {G_2}} {\lambda^* - \epsilon}\,,
\]
where we used that second moments of the gradients are bounded by Assumption~\ref{ass:second-moment} and that the Bregman divergence is $(\lambda^* - \epsilon)^{-1}$-local-strong-convex in $\mathcal O_\epsilon$ by Lemma~\ref{lem:localconv}. Taking expectations conditional on the current iterates, we obtain that
\begin{align*}
    \mathbb E\left[ V_h(\lambda_{t+1},\mu_{t+1}) 
 \mid \lambda_{t},\mu_{t} \right] &\le 
    V_h(\lambda_{t},\mu_{t})
    - \alpha \gr\left(k^{\min}(\lambda_t,\mu_t) \right) \cdot (\lambda_t - \lambda^*) 
    - \alpha \gb\left(k^{\min}(\lambda_t,\mu_t) \right) \cdot (\mu_t - \mu^*)\\
    &+ \alpha^2 \frac{\bar {G_2}} {\lambda^* - \epsilon}\,.
\end{align*}
For the budget constraint, we know that in the set $\mathcal O_\epsilon$ there exists $\underline g > 0$ such that $\gb\left(k^{\min}(\lambda,\mu)\right) \ge \underline g$. Therefore, using that $\mu^* = 0$ and $\mu_t \ge 0$ we obtain that
\[
    \gb\left(k^{\min}(\lambda_t,\mu_t) \right) \cdot (\mu_t - \mu^*) = \gb\left(k^{\min}(\lambda_t,\mu_t) \right) \cdot \mu_t \ge \underline g \cdot \mu_t \ge \underline g  \cdot V_h(\mu^*, \mu_t)\,. 
\]
For the ROS constraint, use that $k^* = 1/\lambda^* + 1$ and $k_t:=k^{\min}(\lambda_t,\mu_t) = (1+\lambda_t)/\lambda_t$ for $(\lambda_t,\mu_t) \in \mathcal O_\epsilon$ to obtain that
\begin{align*}
     \gr\left(k^{\min}(\lambda_t,\mu_t) \right) \cdot (\lambda_t - \lambda^*)
    &=-\gr\left(k_t \right) \cdot (k_t - k^*) \cdot \frac{\lambda^* - \lambda_t}{k_t - k^*} =  -\gr\left(k_t \right) \cdot (k_t - k^*) \cdot \lambda_t \cdot \lambda^*\\
    &\ge \ell (k_t - k^*)^2 \cdot \lambda_t \cdot \lambda^* = \frac{\ell}{\lambda_t \cdot \lambda^*} (\lambda_t - \lambda ^*)^2 \\
    &\ge \frac{\ell}{\lambda^* (\lambda^*+ \epsilon)} \cdot (\lambda - \lambda^*)^2\\
    &\ge \frac{2\ell(\lambda^* - \epsilon)^2}{(\lambda^*)^2 (\lambda^*+ \epsilon)} \cdot V_h(\lambda^*,\lambda)^2\,,
\end{align*}
where the first inequality follows from the strong monotonicity condition in Assumption~\ref{ass:strongly-monotone} and that dual variables are non-negative, the second inequality because $\lambda_t \le \lambda^* + \epsilon$ for all $(\lambda_t,\mu_t) \in \mathcal O_\epsilon$, and  the last inequality follows because the Bregman divergence of the negative entropy function satisfies $V_h(y,x) \le y/(2\min(x,y)^2) (y-x)^2$ for $x,y > 0$ together with $\lambda_t \ge \lambda^* - \epsilon$ for all $(\lambda_t,\mu_t) \in \mathcal O_\epsilon$. Putting everything together, we obtain that there exists constant $C_1, C_2$ such that
\begin{align*}
    \mathbb E\left[ V(\lambda_{t+1},\mu_{t+1}) 
 \mid \lambda_{t},\mu_{t} \right] &\le 
    (1-\alpha C_1) V(\lambda_{t},\mu_{t})
    + \alpha^2 C_2\,.
\end{align*}

We are now ready to invoke the following classical theorem on stochastic stability.

\begin{theorem}[{\citet[p.~86]{kushner1967stochastic}}] Let $x_t, t=1,\ldots,T$ be a Markov process and $V(x)$ a continuous non-negative function with 
\[
    \mathbb E\left[V(x_{t+1}) \mid x_t \right] \le V(x_t) / \beta + \phi
\]
for every $x$ such that $V(x) < m$, where $\beta > 1$ and $\phi \ge 0$. Then
\[ 
    \mathbb P\left\{ \max_{t=1,\ldots,T} V(x_t) \ge m\right\} \le \frac{V(x_1)}{\beta^T m} + \frac{(1-\beta^{-T}) \phi \beta}{(\beta-1)m}\,.
\]
\end{theorem}
Setting $\beta = 1/(1-C_1 \alpha)$ and $\phi = C_2\alpha^2$, we obtain that for $C_1 \alpha < 1$, which holds for large enough $T$
\begin{align*}
     \mathbb P\left\{ \exists t : (\lambda_t, \mu_t) \not\in \mathcal O_\epsilon, \tau < t \le T  \mid (\lambda_\tau, \mu_\tau) \in \mathcal O_\epsilon\right\} &\le \mathbb P\left\{ \max_{t=\tau+1,\ldots,T} V(\lambda_t,\mu_t) \ge m \mid (\lambda_\tau, \mu_\tau) \in \mathcal O_\epsilon \right\} \\
     &\le\beta^{-T} + \frac{(1-\beta^{-T}) \phi \beta}{(\beta-1)m}\\ 
     &\le \exp(-C_1 T\alpha) + \frac {C_2 \alpha}{C_1 m}\,,
\end{align*}
where the last equation follows because $\phi \beta / (\beta-1) = C_2 \alpha /C_1$, $\beta^{-T} = (1-C_1\alpha)^T \le \exp(-C_1 T \alpha)$. Setting $\alpha \approx T^{-1/2}$ we obtain the following result.

\begin{prop}\label{prop:prob-stability}
The algorithm is orbital stable, that is, the probability of leaving the orbit after time $\tau$ is bounded by
\[
    \mathbb P\left\{ \exists t : (\lambda_t, \mu_t) \not\in \mathcal O_\epsilon, \tau < t \le T  \mid (\lambda_\tau, \mu_\tau) \in \mathcal O_\epsilon\right\} = \left(T^{-1/2}\right)\,.
\]
\end{prop}

\subsubsection{Step 3: Regret Analysis}

As before, suppose that the ROS constraint is binding at the optimal solution. Consider an alternate algorithm that (1) behaves as the original algorithm up to time $\tau$ and (2) after time $\tau$ always uses the multiplier of the ROS constraint and projects the dual variable to $[0,\lambda^* + \epsilon]$. Let $\hlambda_t$ be the dual variable in this new algorithm and denote the multiplier used by $\hat k_t = (1+\hlambda_t)/\hlambda_t$. We denote by $E_t = \left\{ k^{\min}(\lambda_t, \mu_t) = \hat k_t\right\}$ the event that the multipliers used by both algorithm match. 

Let  $\taub$ be a stopping time as defined in \cref{defn:Defs} corresponding to the first time the budget is depleted for some initial budget $B=\rho T$. Because values are non-negative, we can lower bound the reward of $\mpacing$ given any $\gseq$ by summing over the value collected only in iterations from $\tau$ up to $\taub$ and conditioning on the event $E_t$
\begin{align*}
\rew(\mpacing,\gseq, \rho)&\geq \sum_{t=\tau}^{\taub} \vt\cdot\xt\left(\vt \cdot k^{\min}(\lambda_t, \mu_t) \right) \cdot E_t\\
&= \sum_{t=\tau}^{\taub} \vt\cdot\xt\left(\vt \cdot (1 +\hlambda_t)/\hlambda_t \right) \cdot E_t\\
&\ge\underbrace{\sum_{t=\tau}^{\taub} \vt\cdot\xt\left(\vt \cdot (1 +\hlambda_t)/\hlambda_t \right)}_{(I)} - \underbrace{\sum_{t=\tau}^{T} v_t \cdot (1 - E_t)}_{(II)}
\end{align*}
where the first equation follows from the definition of the event $E_t$, and the last inequality follows because $x_t \le 1$ and adding back periods after $\taub$. We bound each term at a time.

For the first term, use that the alternate algorithm always bid according to the ROS constraint to write
\[
    \vt\cdot\xt\left(\vt \cdot (1 +\hlambda_t)/\hlambda_t \right)
    = 0\cdot\rho + \ftscombined(\hlambda_t ,0)- \sum_{t=\tau}^{\taub}\hlambda_t \cdot \gr_t(\hlambda_t,0)
\]
Taking expectations, we can use that $\taub$ is a stopping time and a martingale argument to obtain that
\begin{align*}
\E_{\gseq_{\!T}\sim \calp^T}\left[(I)\right]
&\ge \E_{\gseq_{\!T}\sim \calp^T}\left[ \sum_{t=\tau}^{\taub} D(\hlambda_t, 0)  - \sum_{t=\tau}^{\taub}\hlambda_t \cdot \gr_t(\hlambda_t,0) \right]\,,
\end{align*}
where $D(\lambda,\mu)$ is the dual function. Let $\bar\lambda = (\taub+1-\tau)^{-1}  \sum_{t=\tau}^{\taub}\hlambda_t$ be the average dual variable for the ROS constraint. Using the convexity of the dual function we obtain that
\begin{align*}
     \sum_{t=\tau}^{\taub} D(\hlambda_t, 0) &\ge (\taub+1-\tau) D(\bar \lambda, 0) 
     \ge \E_{\gseq\sim \calp^T} \left[ \rew(\opt,\gseq) \right] - O\left(T^{1/2}\right)\,,
\end{align*}
where we used that $(\bar \lambda,0)$ is dual feasible and weak duality together with $\tau = O\left(T^{1/2}\right)$ and $T - \taub = O\left(T^{1/2}\right)$. Because the alternate algorithm projects dual variables to $[0,\lambda^*+\epsilon]$, Lemma~\ref{lem:localconv} implies that the Bregman divergence of the generalized negative entropy is $1/(\lambda^*+\epsilon)$-strongly convex. Applying the mirror descent guarantee in Lemma~\ref{lem:onlineMD} to the linear functions $w_t(\lambda) = \lambda \cdot \gr_t(\hlambda_t,0)$ we obtain that
\[
    \sum_{t=\tau}^{\taub}\hlambda_t \cdot \gr_t(\hlambda_t,0)
    = \sum_{t=\tau}^{\taub} w_t(\hlambda_t) - w_t(0) = O\left(T^{1/2}\right)\,,
\]
because the alternate algorithm updates the dual variable of the ROS constraint according to $\gr_t(\hlambda_t,0)$. Therefore, we have that
\begin{align}\label{eq:first-term}
\E_{\gseq_{\!T}\sim \calp^T}\left[(I)\right]
&\ge \E_{\gseq\sim \calp^T} \left[ \rew(\opt,\gseq) \right] - O\left(T^{1/2}\right)\,.
\end{align}

For the second term, using that values are independent of the event $E_t$ we obtain
\begin{align}
    \E_{\gseq_{\!T}\sim \calp^T}\left[(II)\right] &= 
    \sum_{t=\tau}^{T} \E[v_t] \cdot \mathbb P\left\{ E_t^\complement \right\}\nonumber\\
    &\le T \cdot \E[v] \cdot \mathbb P\left\{ \cup_{t=\tau}^T E_t^\complement \right\}\nonumber\\
    &= T \cdot \E[v] \cdot \left( \mathbb P\left\{ \cup_{t=\tau}^T E_t^\complement \mid A \right\} \mathbb P\left\{ A \right\} + \mathbb P\left\{ \cup_{t=\tau}^T E_t^\complement \mid A^\complement \right\} \mathbb P\left\{ A^\complement \right\} \right)\nonumber\\
    &\le T \cdot \E[v] \cdot \left( \mathbb P\left\{  \exists t : (\lambda_t, \mu_t) \not\in \mathcal O_\epsilon, \tau < t \le T \mid (\lambda_\tau, \mu_\tau) \in \mathcal O_\epsilon\right\} + \mathbb P\left\{ (\lambda_\tau, \mu_\tau) \not\in \mathcal O_\epsilon \right\} \right)\nonumber\\
    &=O\left(T^{1/2}\right) \label{eq:second-term}
    \,,
\end{align}
where the first inequality follows because values are i.i.d.~and $\mathbb P\{ E_t^\complement \} \le \mathbb P\{ \cup_{t=\tau}^T E_t^\complement \}$ for all $t=\tau,\ldots,T$, the second equality follows from conditioning on the event $A = \{ (\lambda_\tau, \mu_\tau) \in \mathcal O_\epsilon\}$, the second inequality follows because probabilities are at most one and if for some $t$ the event $E_t^\complement$ is true then it must be the case that $(\lambda_t,\mu_t) \not\in \mathcal O_\epsilon$ since the algorithm bids according to the ROS multiplier in the orbit of the optimal dual solution, and the last inequality follows from Proposition~\ref{prop:prob-identification} and Proposition~\ref{prop:prob-stability}.

Combining \eqref{eq:first-term} and \eqref{eq:second-term} we conclude that
\[
    \reg(\mpacing,\calp^T) = O\left(T^{1/2}\right)\,.
\]
\subsection{Online Mirror Descent Results}\label[app]{sec:omd}
The following are some known results of Online Mirror Descent that we used in our previous analysis.
\begin{lem}[\cite{bubeck2015convex}, Theorem $4.2$]\label{lem:onlineMD}
Let $h$ be a mirror map which is $\rho$-strongly convex on $\mathcal{X}\cap \mathcal{D}$ with respect to a norm $\|{}\cdot{}\|$. Let $f$ be convex and $L$-Lipschitz with respect to $\|{}\cdot{}\|$. Then, mirror descent with step size $\alpha$ satisfies \[ \sum_{s=1}^t \left( f(x_s) - f(x) \right) \leq \frac{1}{\alpha} V_h(x, x_1) + \alpha \frac{L^2t}{2\rho}.\] 
\end{lem}

\begin{lem}[\cite{allen2014using}]
\label{lem:localconv}
The Bregman divergence of the generalized negative entropy satisfies ``local strong convexity'': for any $x, y>0$, \[ V_h(y, x) = y \log(y/x) + x- y\geq \frac{1}{2\max(x,y)}\cdot (y-x)^2.\] 
\end{lem}
\begin{proof}
The claimed inequality is equivalent to \[t \log t  \geq (t-1) + \frac{1}{2\max(1,t)}\cdot (t-1)^2\numberthis\label[ineq]{eq:initLSCineq}\] for $ t >0$. 
Suppose $t\geq 1$. Then, choosing $u = 1-1/t$, \cref{eq:initLSCineq} is equivalent to 
\[-\log (1-u) \geq u + \frac{1}{2} u^2, \] for $u \in [0, 1)$, which holds by Taylor series. Suppose $0 < t \leq 1$. Then \cref{eq:initLSCineq}  is equivalent to \[ \log t - \frac{1}{2}\left( t - \frac{1}{t}\right)\geq 0,\] which may be checked by observing that the function is decreasing and equals zero at $t=1$. This completes the proof of the claim. 
\end{proof}
\section{Proofs of Theorem \ref{thm:constaints}}\label{app:constraints}

To prove the result, we first show the next lemma, which bounds the constraint violations for time $t$.
\begin{lem}\label{prop:gradientProperties}
Recall $\gr_t = v_t \cdot x_t\left( b_t \right) - p_t\left( b_t\right)$ and $\gb_t = \rho - p_t\left( b_t\right)$ with $\lambda_t,\mu_t$ being the dual variables for the ROS and the budget constraint  respectively, and $b_t = v_t \cdot k^{\min}(\lambda_t, \mu_t)=v_t\cdot \min\left\{1/\mu_t,1+1/\lambda_t \right\}$ being the bid used by the algorithm. If the payment and allocation functions satisfy $0\leq p_t(b)\leq b\cdot x_t(b_t)$ for any bid $b>0$ (e.g. truthful auctions), then we have
\[
\gr_t \geq -\frac{1}{\lambda_t} \quad \text{and} \quad \gb_t \geq \rho -\frac{1}{\mu_t}
\]
\end{lem}
\begin{proof} Our condition only says the payment is always non-negative and at most the bid. Recall we also normalize the functions so that $v_t$, $p_t$ and $x_t$ all have range $[0,1]$. For the ROS constraint, since $b_t\leq \v_t\cdot (1+1/\lambda_t)$ and $p_t(b_t)\leq b_t\cdot x_t(b_t)$, we get
\[
\gr_t \geq (v_t-b_t)\cdot x_t(b_t) \geq -\frac{v_t}{\lambda_t} \cdot x_t(b_t) \geq - \frac{1}{\lambda_t}. 
\]
Similarly, for the budget constraint because $b_t\leq \v_t/\mu_t$, we get
\[
\gb_t \geq \rho -b_t\cdot x_t(b_t) \geq \rho -\frac{v_t}{\mu_t} \cdot x_t(b_t) \geq \rho - \frac{1}{\mu_t}. 
\]
\end{proof}

The first result in Theorem \ref{thm:constaints} on ROS constraint violation can be obtained from the below lemma.
\begin{lem}
\label{lem:ros_constraint_violation}
Consider a run of the min pacing algorithm starting at $\lambda_1>0$ and $\alpha =\frac{1}{\sqrt{T}}$, then for any outcome $\gseq$ over the $T$ iterations, the ROS constraint violation satisfies \[\sum_{t=1}^T p_t(b_t) -v_t \cdot x_t({b}_t) =-\sum_{t = 1}^T \gr_t \leq 2\sqrt{T}\log \frac{T}{\lambda_1}.\] 
\end{lem}
\begin{proof}
Equation~\eqref{eq:ros-dual-update} in the algorithm implies $\lambda_{t+1} =\exp\left[-\alpha\sum_{t'=1}^{t}\gr_{t'}\right]$. If $-\sum_{t = 1}^T \gr_t \leq \sqrt{T}\log\frac{T}{\lambda_1}$, we are done. Otherwise, let $T^{\prime}$ be the last time that $-\sum_{t = 1}^{T^{\prime}} \gr_t \leq \sqrt{T}\log \frac{T}{\lambda_1}$, so we know for any $t>T^{\prime}$, the dual variable $\lambda_t$ must be larger than $T$ since 
\[
\lambda_t =\lambda_1\cdot\exp\left[-\alpha\sum_{t'=1}^{t}\gr_{t'}\right]>\lambda_1\cdot \exp\left[\alpha\sqrt{T}\log \frac{T}{\lambda_1}\right] = T
\]
By~\cref{prop:gradientProperties} we know $\gr_t \geq -\frac{1}{\lambda_t}$, so $
-\gr_t\leq \frac{1}{\lambda_t}\leq \frac{1}{T}$ for all the iterations $t$ after $T'$. Since there are at most $T$ such iterations, we get
\[ -\sum_{t=1}^T \gr_t = -\sum_{t=1}^{T^{\prime}}\gr_t - \sum_{t > T^{\prime}}\gr_t\leq \sqrt{T}\log \frac{T}{\lambda_1} + 1\leq 2\sqrt{T}\log \frac{T}{\lambda_1}. \] 
\end{proof}

The second result in Theorem \ref{thm:constaints} on stopping time can be obtained from the below lemma.
\begin{lem}
\label{lem:bud_constraint_violation}
Let $\mu^{\max}=1/\rho + 1$, and consider a run of the min pacing algorithm starting at $\mu_1\in (0,\mu^{\max}]$ and $\eta =\frac{1}{\sqrt{T}}$, then for any outcome $\gseq$ over the $T$ iterations, we have $\mu_t\leq \mu^{\max}$ for all $t\leq \taub$, and $T-\taub\leq \frac{\sqrt{T}}{\rho}\cdot\log\frac{10\mu^{\max}}{\mu_1}=O(\sqrt{T})$
\end{lem}
\begin{proof}
The part of $\mu_t\leq \mu^{\max}$ follows inductively. If $\mu_t\leq \mu^{\max}$, either $\mu_t\leq 1/\rho$, then since the step-size $\eta$ is chosen to be small enough we have $\mu_{t+1}\leq 1/\rho + 1$, otherwise if $\mu_t> 1/\rho$, by ~\cref{prop:gradientProperties} we know $\gb_t>0$ and thus $\mu_{t+1}\leq \mu_{t}\leq \mu^{\max}$.

The part of $\taub$ can be shown by contradiction. Suppose $\taub<T-\frac{\sqrt{T}}{\rho}\cdot\log\frac{10\mu^{\max}}{\mu_1}$, it means $\sum_{t=1}^{\taub-1} p_t(b_t)\geq \rho\cdot T-2$  and thus 
\[
\sum_{t=1}^{\taub-1} \gb_t \leq \rho\cdot\taub - (\rho\cdot T-2)\leq - \sqrt{T}\cdot\log\frac{10\mu^{\max}}{\mu_1} + 2.
\]
Similar to the ROS case, note $\mu_{\taub}=\mu_1\cdot \exp\left[-\eta\sum_{t=1}^{\taub-1}\gb_{t}\right] \geq \mu^{max}$, which gives a contradiction.
\end{proof}
\section{Analysis of sequential algorithm }\label{app:sequential_fails}
We prove Proposition~\ref{prop:sequencial} in this section. That is, we will show that for any initialization of the sequential pacing algorithm, i.e. choice of initial values $\mu_0,\lambda_0$ of the dual variables and their respective step-sizes $\eta,\alpha$, there will always be some instance on which the algorithm performs poorly, i.e. it either violates the ROS constraint by at least $\Omega(T)$ or has a regret at least $\Omega(T)$. 

Without loss of generality, we assume $\mu_0$ and $\lambda_0$ are both $O(1)$. All the instances we use in the proof will be deterministic, i.e. $v,x(\cdot),p(\cdot)$ are drawn i.i.d from a point distribution. In particular, all instances we consider have fixed values $v_t=1$, $x_t(b)=\min(\frac{b}{4},1)$ and $p_t(b)=\min(\frac{b^2}{8},2)$ for all $b\geq 0$. Effectively the bid ranges from $0$ to $4$, and is equivalent to the bid multiplier as $v=1$. We pick these values for notation simplicity, and it is easy to scale all quantities down to satisfy our model where $v,x,p$ are all in $[0,1]$. Note that the payment function $p$ is the truthful pricing corresponding to the allocation function $x$ in our example. We start with the following observations for our instance.
\begin{observation}
    \label{obs:even-pacing}
It is straightforward to see that in each iteration, the value is a concave function on the payment, i.e. $(v\cdot x)=\sqrt{p/2}$ (Figure\ref{fig:sequential_proof}), and thus if we fix some total spend $P$ over some $t$ iterations, the largest total value is achieved by spending evenly (i.e. $P/t$) in each of the $t$ iterations. Similarly because of concavity, if there is an additional constraint that the per-iteration spend is at least $l\geq P/t$, the optimal total value is achieved by spending $l$ per-iteration (over any $P/l<t$ iterations).
\end{observation} 
\begin{observation}
    \label{obs:max-slack}
In each iteration, the largest ROS slack one can achieve is at most $1/8$, i.e., $\max_b \left\{ v\cdot x(b)-p(b)\right\} = 1/8$ by bidding $b=1$.
\end{observation}

\begin{figure}[h]
    \centering
    \resizebox{0.4\textwidth}{!}{
    \begin{tikzpicture}[scale=1,every node/.style={scale=1}]
        \begin{axis}[xlabel={spend},ylabel={conversion value}, ytick={0}, xtick={0},extra x ticks={2},extra y ticks={1}, legend pos=outer north east, legend cell align=left, xmin=0, xmax=2.2, ymin=0, ymax=1.2,axis x line = middle,axis y line = left]       
        \addplot[domain=0:2,color=black,line width=1pt,samples=100]{(x/2)^0.5};
        \addplot[domain=0:2,color=black,dashed,line width=1pt,samples=2]{x};        
        \legend{Landscape,RoS constraint}
        \end{axis}
    \end{tikzpicture}
    }
    \caption{Achievable spend vs conversion value for the sequential example. }
    \label{fig:sequential_proof}

\end{figure}
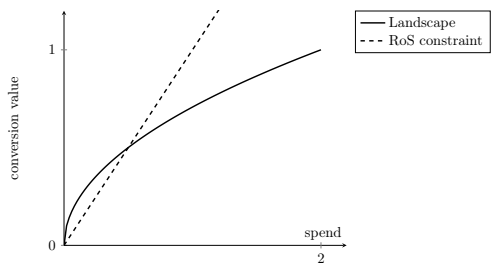

Fix any $\mu_0,\lambda_0,\eta,\alpha$, we will consider a pair of instances. The first instance $\hat{\mathcal{I}}$ has $\hat{\rho}=1.9$ (and the $v,x,p$ as described above). It is easy to see for $\hat{\mathcal{I}}$ that $k^*=\kr=2$ and $\kb=\sqrt{8\cdot 1.9}>2$, and it is optimal to spend $0.5$ per iteration and get $T\cdot v\cdot \frac{2}{4}=T/2$ total value. It is also easy to check that this instance satisfies all the assumptions we need for the min-pacing algorithm. Consider the sequential pacing algorithm with two cases
\begin{enumerate}
    \item If the total spend over $T$ iterations is at least $P\geq 0.6\cdot T$. The maximum total value in this case is achieved by spending $P/T$ per iteration (Observation~\ref{obs:even-pacing}), which means bidding $b=\sqrt{8P/T}$ and get value $\sqrt{P/(2T)}$. Thus the total value is at most $\sqrt{P\cdot T/2}$, so the total ROS constraint violation is at least 
    \[
    P-\sqrt{P\cdot T/2} = (P/T - \sqrt{\frac{P}{2T}})\cdot T.
    \]
    It is easy to check this is at least $\Omega(T)$ when $P\geq 0.6\cdot T$, so the ROS constraint violation would be linear in $T$.
    \item If the total spend over $T$ iterations is at most $0.6\cdot T$. Consider any iteration after the first $0.7\cdot T$ iterations, and we know $\mu_t = \mu_0\exp\Big(-\eta\cdot(\hat{\rho}\cdot t -\sum_{t'<t}p_{t'}(b_{t'})\Big)\leq \mu_0\exp\Big(-\eta\cdot 0.73\cdot T\Big)$ for any $t\geq 0.7T$ since $\hat{\rho}=1.9$ and total spend is at most $0.6\cdot T$. There are two sub-cases:
    \begin{itemize}
        \item If $\mu_0\exp\Big(-\eta\cdot 0.73\cdot T\Big)\leq 1/3$, we know $\mu_t\leq 1/3$ and thus $b_t = \frac{1}{\mu_t}\cdot \frac{\lambda_t+1}{\lambda_t}\geq 3$ for all $t\geq 0.7\cdot T$, which means we will have a per iteration ROS violation of at least $9/8-3/4 = 0.375$ (with $b_t=3$) in each of the last $0.3\cdot T$ iterations, since the ROS violation increases with $b$ over the region $b\geq 3$. In each of the first $0.7\cdot T$ iterations, the ROS slack we can gain is at most $1/8$ (Observation~\ref{obs:max-slack}), so the total ROS constraint violation is at least $0.375\cdot 0.3 \cdot T - 0.7\cdot T / 8 \geq 0.025\cdot T$. 
        \item If $\mu_0\exp\Big(-\eta\cdot 0.73\cdot T\Big)> 1/3$, then we know $\mu_0>1/3$ and $\eta \leq \frac{\ln(3\mu_0)}{0.73\cdot T}$. 
        
    \end{itemize}
\end{enumerate}
We can conclude from the above discussion that the only possible scenario where an instantiation of the sequential pacing algorithm won't incur a $\Omega(T)$ violation of the ROS constraint on the instance $\hat{\mathcal{I}}$ is in the last sub-case, which means the step-size $\eta$ of the budget dual variable is $O(1/T)$. If that is the case, it is easy to see such an instantiation must perform poorly on a budget-binding instance when we need $\eta$ to be large so the budget dual variable $\mu$ can increase fast enough to lower the bid sufficiently. 

More specifically, when $\mu_0\exp\Big(-\eta\cdot 0.73\cdot T\Big)> 1/3$ holds, we consider the instance $\tilde{\mathcal{I}}$ with $\tilde{\rho}=\frac{1}{200\mu_0^4}$. Note $\tilde{\rho}$ is $\Theta(1)$ since we assume $\mu_0$ is $O(1)$ and in this case $\mu_0>1/3$. We have $k^*=\kb=\frac{1}{5\mu_0^2}<1.8$, $\kr=2$, and the maximum total value is $\frac{T}{20\mu_0^2}$ achieved by bidding $\kb$ and spending $\tilde{\rho}$ in each iteration. Since the pacing algorithm guarantees to obey the budget constraint (by not bidding above the remaining budget at any time), the total spend is at most $\tilde{\rho}\cdot T = \frac{T}{200\mu_0^4}\leq 0.73\cdot T$ (as $\mu_0>1/3$), so we know that $\sum_{t'\leq t}\tilde{\rho}-p_{t'}(b_{t'})\geq -\sum_{t'\leq t}p_{t'}(b_{t'}) \geq -0.73\cdot T$ for any $t$. Thus
\[
\mu_t = \mu_0\exp\Big(-\eta\cdot \sum_{t'\leq t}\left(\tilde{\rho}-p_{t'}(b_{t'})\right) \Big)\leq \mu_0\exp\Big(\eta\cdot 0.73\cdot T\Big) < \mu_0\cdot (3\mu_0)=3\mu_0^2,
\]
where the last inequality follows from $\mu_0\cdot \exp\Big(-\eta\cdot 0.73\cdot T\Big)> 1/3$. Consequently, we have $b_t = \frac{1}{\mu_t}\cdot \frac{\lambda_t+1}{\lambda_t}\geq \frac{1}{3\mu_0^2}$ and thus spend at least $\frac{1}{72\mu_0^4}$ in all iterations before the budget is depleted. It is straightforward to see that the maximum possible total value under this condition is obtained when bidding exactly $\frac{1}{3\mu_0^2}$ (and spending exactly $\frac{1}{72\mu_0^4}$) per iteration until the budget depletes (Observation~\ref{obs:even-pacing}). This gives a value of $\frac{1}{12\mu_0^2}$ per iteration, and the budget is depleted after $\tilde{\rho}\cdot T/\left(\frac{1}{72\mu_0^4}\right)=\left(\frac{T}{200\mu_0^4}\right)/\left(\frac{1}{72\mu_0^4}\right)=\frac{72\cdot T}{200}$ iterations. The total value obtained by sequential pacing in this case is at most $\frac{3T}{100\mu_0^2}$, which is at least $\Omega(T)$ smaller than the optimal value of $\frac{T}{20\mu_0^2}$ (as $\mu_0$ is $O(1)$ by assumption). 

This completes our argument that given any instantiation of the sequential pacing algorithm, there exists an instance, which satisfies all the assumptions we need for the min pacing algorithm, such that the sequential pacing algorithm either incurs at least $\Omega(T)$ violation of the ROS constraint, or has a regret at least $\Omega(T)$.

\section{Supplementary material for Empirical study}
\label{app:empirical}
\subsection{Semi-synthetic Dataset Construction}
\label{app:empirical-dataset}
\paragraph{Generative Model. } We will use a i.i.d. stochastic model to generate the $x_t(b_t)$ and $p_t(b_t)$ in iteration $t$ as a function of bid $b_t$ (as discussed earlier, we slightly abuse notation to use $b_t$ to be the multiplier to $tcpa \cdot pconv$). We use a Poisson distribution for $x_t(b_t)$ (i.e. number of clicks in an iteration at a bid $b_t$). With parameter $\lambda$, its probability mass function is $f(x;\lambda) = \frac{\lambda^x e^{-\lambda}}{x!}$. The parameter $\lambda$ is the expected number of clicks in an iteration, and we set it using the bidding landscape. In particular, for the model corresponding to a campaign $C$, the expected number of clicks at bid $b_t$ in an iteration would be $\lambda_C(b_t)=click_{C}(b_t)/T$ where $T$ is the total number of iterations in a day. In our empirical evaluation, we pick $T=144$, which translates to each iteration being a $10$-minute period, i.e., the dual variables of the algorithms are updated every 10 minutes instead of after every auction. We also derive from the bidding landscape a cost-per-click $cpc_C(b_t)=\frac{cost_{C}(b_t)}{click_{C}(b_t)}$.

To summarize, the value, click and cost at bid $b_t$ are as follows:

$\boxed{
v_t = tcpa(C)\cdot pconv_t(C), \qquad 
x_t(b_t) \sim  Poisson(\lambda_C(b_t)) 
,\qquad p_t(b_t) = x_t(b_t)\cdot cpc(b_t) \cdot noise_p}$

where we introduce i.i.d. non-negative multiplicative noise 
$noise_p$ with expected value $1$ to the cost. In our evaluation, we use a Gaussian distribution centered at $1$ with standard deviation $0.1$ and truncated to be within $[0,2]$ (so it's non-negative and has expected value $1$). Also, when empirically evaluating the tCPA campaigns, the conversion rates $pconv_t(C)$ are drawn from a Gaussian distribution centered at the average $pconv$ of the campaign with a standard deviation of $0.1$ and truncated to be in $[0,2]$.

It is easy to see that in the stochastic i.i.d. model, the optimal value of~\eqref{eq:obj_relaxation} is an upper bound on the expectation of the ex-post optimal value. In our generative model, by design we have

\paragraph{Bidding Landscape. } To see how the auction performance of a campaign determines its model parameters in the generative model, it is useful to begin with the notion of a bidding landscape. For each campaign $C$, we construct a bidding landscape as a function from bids to the (predicted) number of clicks and cost. This is done first at a per-query level using auction simulation. In more detail, for an ad opportunity (a.k.a. query) $q$ where campaign $C$ is eligible to show its advertisement, we look at the logged bids of all the other campaigns participating in the auction for this query $q$, and simulate the auction for any bid $b$ of $C$ to know if/where $C$'s ad would be shown. This gives us the predicted number of clicks and cost per click corresponding to any particular bid $b$, and we refer to them as  $click_{C,q}(b)$ and $cost_{C,q}(b)$. In our model, we use the actual (i.e., advertiser submitted) target cost per acquisition of $C$ as $tcpa(C)$, and the logged average predicted conversion probability generated by the production machine-learning model as $pconv_q(C)$. 
For a query $q$, bids are given by $b = k \cdot v_q$ where $v_q = tcpa (C) \cdot pconv_q(C)$ is the value of the query.

We aggregate these single-query landscape functions to get $C$'s daily bidding landscape by summing up the respective functions over all the queries in a day, e.g., $click_{C}(k)=\sum_{q}click_{C,q}(k \cdot v_q)$, and $cost_{C}(k)=\sum_{q}cost_{C,q}(k \cdot v_q)$. Note that these functions are non-decreasing in $k$. The per-query bidding landscapes are inherently step functions represented by the various bid thresholds that makes $C$'s ad to be displayed at various positions (or not displayed at all). While the aggregated landscapes are already smoother than the per-query landscapes, we further smooth the aggregated landscapes by linearly interpolating between consecutive thresholds. See Figure~\ref{fig:landscape} for an example of the aggregated daily bidding landscape of an ad campaign.\footnote{We normalize the values of click, value and cost in all the plots of this section, so the quantities shown do not represent real traffic or revenue.}

\begin{figure}[h]
  \begin{center}  
    \includegraphics[width=0.75\textwidth]{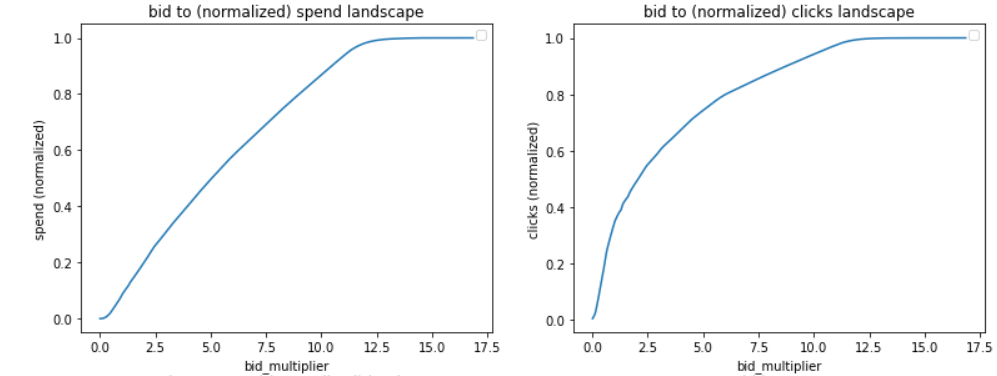}
    \caption{The bidding landscape of an example campaign. The x-axis is the bid (as a multiplier to value), and the y-axis are the daily cost and number of clicks (all normalized to be in $[0,1]$) respectively.}
    \label{fig:landscape}
  \end{center}
\end{figure}

\begin{figure}
	\centering
	\begin{subfigure}{\textwidth}
  		\centering
  		\includegraphics[width=\textwidth,height=6cm]{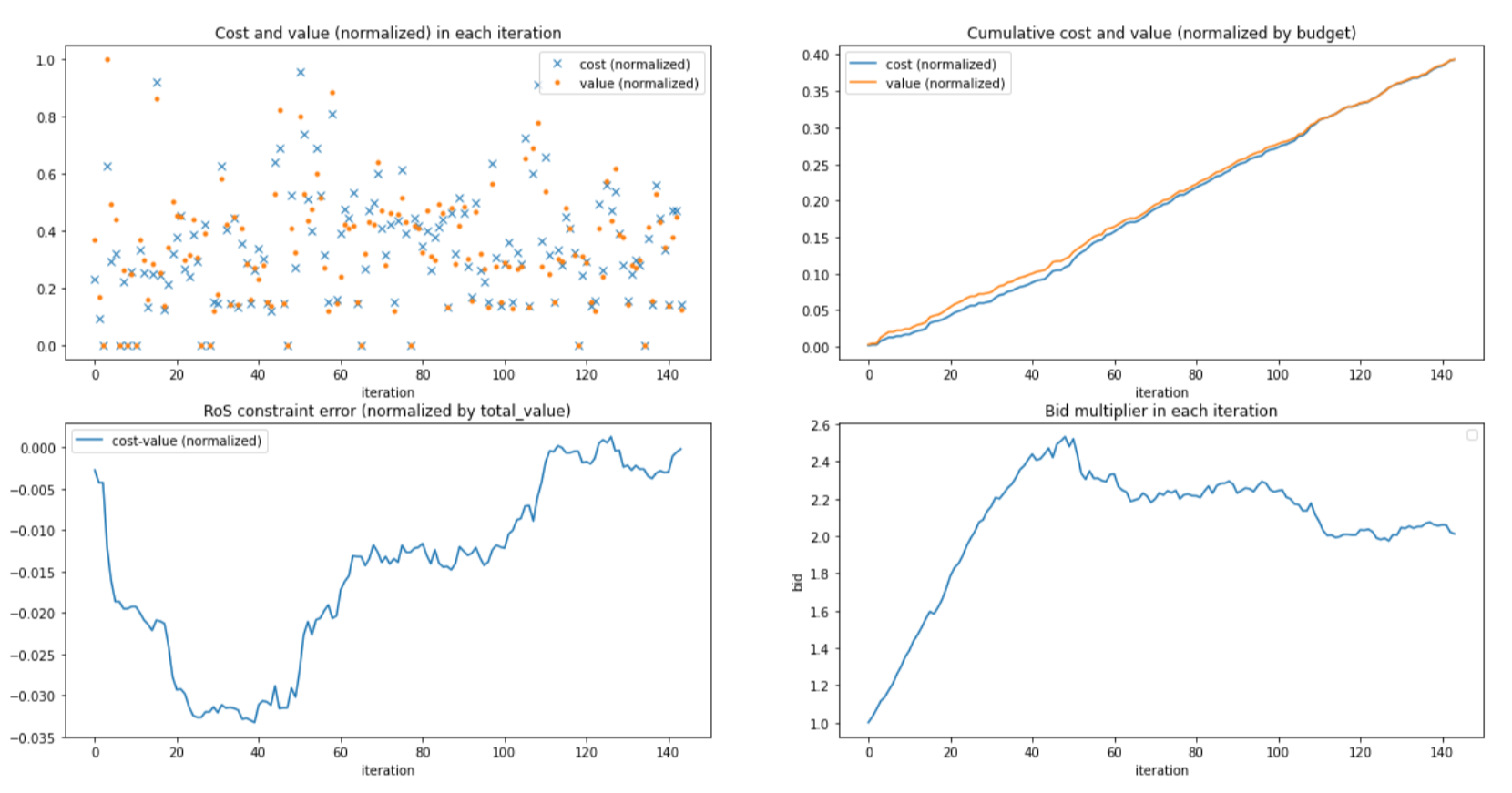}
  		\caption{Dual-optimal Pacing Algorithm}
  		\label{fig:joint_example}
	\end{subfigure}
 	\begin{subfigure}{\textwidth}
  		\centering
  		\includegraphics[width=\textwidth,height=6cm]{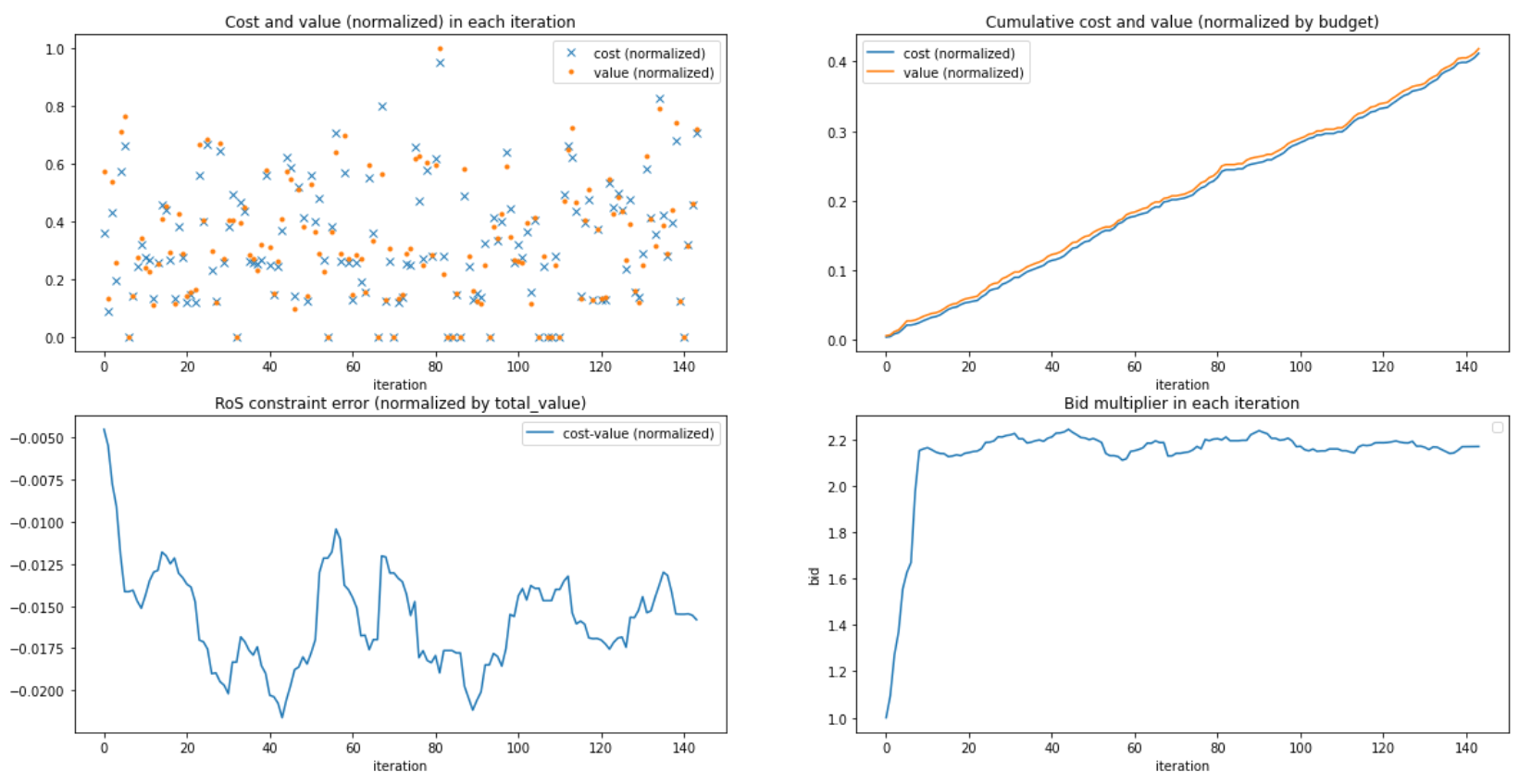}
  		\caption{Min Pacing Algorithm}
  		\label{fig:min_example}
	\end{subfigure}
	\begin{subfigure}{\textwidth}
  		\centering
  		\includegraphics[width=\textwidth,height=6cm]{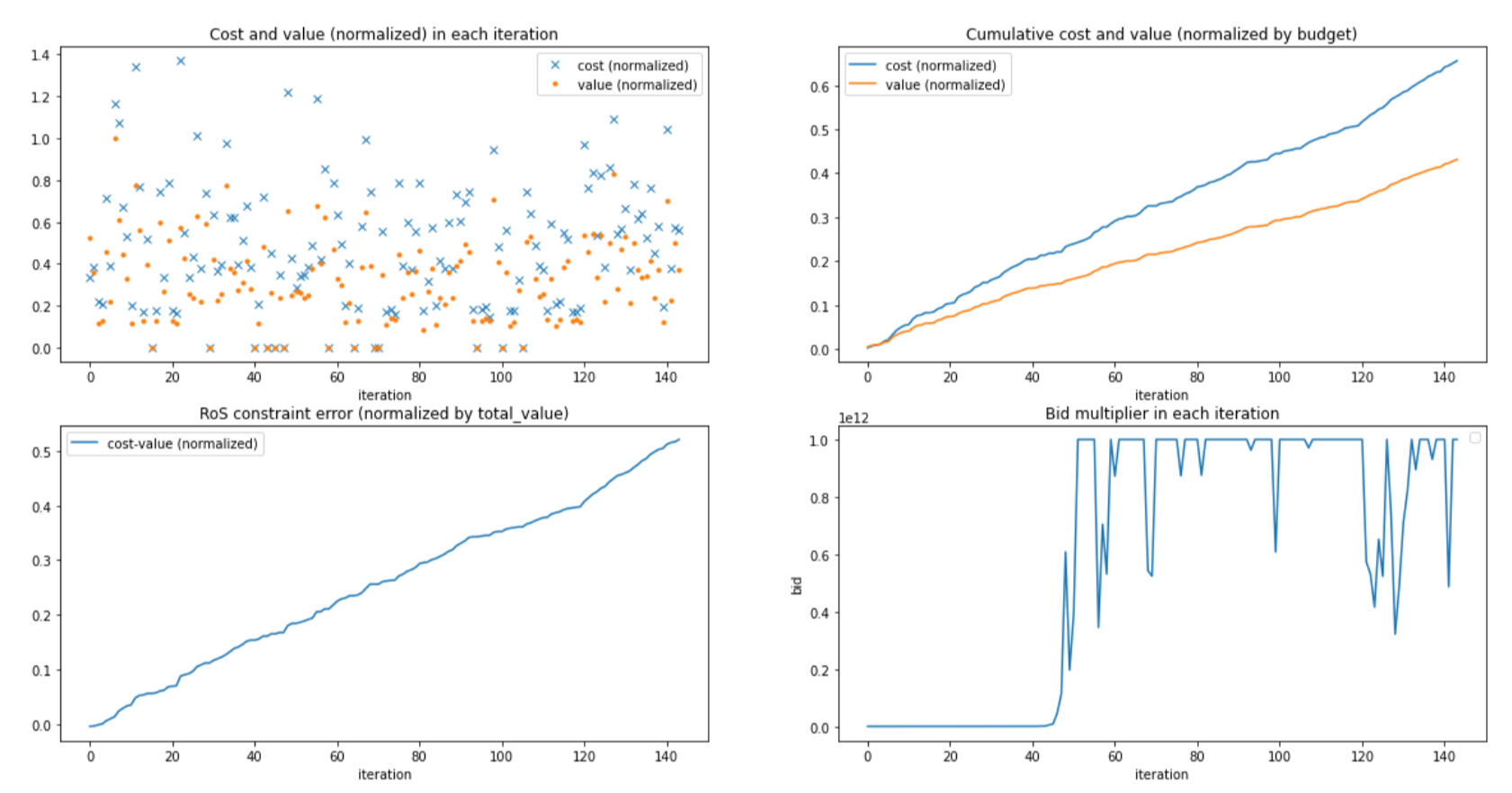}
  		\caption{Sequential Pacing Algorithm}
  		\label{fig:seq_example}
	\end{subfigure}
	\caption{Simulation of the dual-optimal bidding (top),  min bidding (middle) and sequential bidding algorithms (bottom) on an example campaign. We plot the per-iteration $value$ and $cost$ (normalized so $value\in[0,1]$), cumulative $value$ and $cost$ (normalized by budget), cumulative ROS error (as $cost-value$ normalized by total $value$), and bids $k$.} 
\label{fig:example-execution}
\end{figure}
\subsection{Empirical Evaluation}
\label{app:empirical-eval}
In our empirical study, we evaluate an algorithm on a campaign as follows. In each iteration, after the algorithm gives the bid it wants to submit, we compute the allocation and payment $x_t,p_t$ using our generative model to get the number of clicks and cost of that iteration, and let the algorithm update the bid for the next iteration. We sum up the total cost and value through all $T$ iterations. For the budget constraint, we follow the common practice to always strictly enforce it as follows: if in an iteration the generated cost is larger than the remaining budget, we modify that iteration's cost and value both to be $0$. We do not enforce the ROS constraint strictly\footnote{Note that it is always possible for an ROS constraint to be temporarily violated after $t$ rounds, but in the $t+1$-th round it could become satisfied because of a really high value query coming through at low cost. Therefore it is suboptimal to stop serving right after ROS constraint gets violated in a round. This is not the case for budget constraint: once violated, it always remains violated because cumulative spend is monotonically increasing.}, but of course, measuring how much the different algorithms violate the ROS constraint is an important aspect of this study and will be discussed here. In Figure~\ref{fig:example-execution} we visualize the pacing algorithms on an example campaign.

In our empirical evaluation, for each campaign, we simulate an algorithm $10$ times to take the average total $\spend$ and total $\conv$ or conversion value as the result of the algorithm on that campaign. For each algorithm, we take the $10,000$ $(\spend,\conv)$ pairs from all the campaigns, and arrange them into buckets based on the relative ROS constraint error\footnote{ROS constraint states that $\spend \leq \conv$. So a constraint violation would imply $\spend > \conv$, i.e., $\spend/\conv-1 > 0$.} $\max\left(0,\spend/\conv - 1\right)$. For each bucket, we sum up the $\conv$ of all the campaigns in it. Moreover, for each algorithm, we do a grid search over the step-sizes used in the dual variables' updates. Each pair of step-sizes (one for each dual variable) is evaluated over the entire dataset, and for each algorithm we pick the best pair of step-sizes according to the total $\conv$ in the bucket of zero ROS constraint error. We compare the results associated with the best step-sizes for each algorithm.
\subsection{Benchmark}
\label{app:empirical-benchmark}
For each campaign, our benchmark (i.e.~\eqref{eq:obj_relaxation}, copied below) is the fluid relaxation of~\eqref{eq:obj_x}, but restricted to uniform bidding, i.e., $b_t = k\cdot v_t$ for all $t$. 
\begin{equation*}
\begin{array}{ll}
\underset{k \ge 0}{\mbox{maximize}} & \sum_{t=1}^T \ex[v_{t}\cdot x_t(k
\cdot v_t)]\\
\mbox{subject to } &  \sum_{t=1}^T \ex[p_t(k\cdot v_t)] \leq  \sum_{t=1}^T \ex[v_{t}x_t(k\cdot v_t)],
\\
& \sum_{t=1}^T \ex[p_t(k \cdot v_t)] \leq \rho T\,.
\end{array}
\end{equation*}


It is easy to see that in the stochastic i.i.d. model, the optimal value of~\eqref{eq:obj_relaxation} is an upper bound on the expectation of the ex-post optimal value. In our generative model, by design we have
\begin{align*}
  \conv_C(k) &= \ex[ v_{t} x_t(b_t)]=\frac{tcpa(C)} T\cdot \ex[pconv_t(C)\cdot click_C\left(k \cdot tcpa(C)\cdot pconv_t(C)\right)]\,,\\
\spend_C(k) &= \ex[p_t(b_t)]=\ex\left[x_t(b_t) \frac{cost_C(b_t)}{click_C(b_t)}\right]\ex[noise_p]=\frac{1}{T} \ex[cost_C\left(k \cdot tcpa(C)\cdot pconv_t(C)\right)]\,,  
\end{align*}
where the expectation is taken with respect to the distribution of conversion probabilities of the different queries.

Our benchmark for campaign $C$ in~\eqref{eq:obj_relaxation} becomes
\begin{equation}
\label{eq:benchmark}
\begin{array}{ll}
\underset{k\ge0}{\mbox{maximize}} &  \conv_C(k)\\
\mbox{subject to } &  \spend_C(k) \leq  \conv_C(k),
\\
& \spend_C(k) \leq \rho\,.
\end{array}
\end{equation}
In our experiments, we approximate $\conv_C(k)$ and $\spend_C(k)$ by performing a certainty equivalent approximation in which we replace random quantities (i.e., the predicted conversion probabilities) by their expected values.
We solve the above optimization problem on the bidding landscape functions by finding the largest bid multiplier $k^*$ such that $\spend_{C}(k^*)$ is below $C$'s budget and $\conv_C(k^*)\geq \spend_{C}(k^*)$. Such multiplier $k^*$ is easy to find using a line search since our landscape functions are all monotone in $k$. Furthermore, the restriction to uniform bidding in ~\eqref{eq:obj_relaxation} is without loss of generality when the $\conv_C(k)$ versus $\spend_C(k)$ function is concave, which qualitatively holds in our data (e.g. Figure~\ref{fig:opt_exp}).


We use $\conv_C(k^*)$ as computed above as the benchmark for $C$ (see Figure~\ref{fig:opt_exp} for examples). Note this captures the expected optimal solution, but algorithms running on the generative model of $C$ may achieve better ex-post value than the benchmark due to the stochasticity of the model. We add up the expected optimal value over all campaigns as the overall benchmark. Figure~\ref{fig:opt_exp} shows the pairs of spend and conversion value levels that can be achieved by varying the bidding multiplier $k$ for a typical campaign. The achievable curves $\left( \spend_C(k), \conv_C(k)\right)_{k \ge 0}$ lie in $\mathbb R_+^2$, start at the origin for $k=0$, increase along both axis as the bid multiplier increases, and end at $k \rightarrow \infty$.
\begin{figure}[h]
  \begin{center}  
    \begin{subfigure}{.5\textwidth}
        \centering
        \includegraphics[width=.8\linewidth]{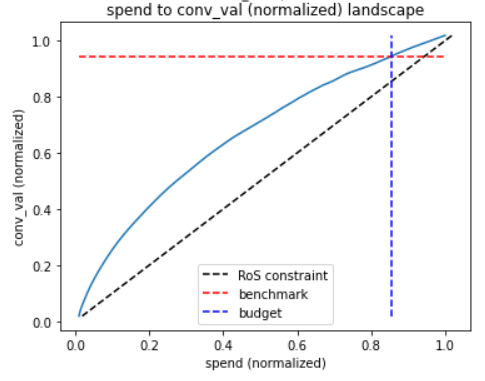}
        \caption{Budget constraint is binding.}
    \end{subfigure}%
    \begin{subfigure}{.5\textwidth}
        \centering
        \includegraphics[width=.8\linewidth]{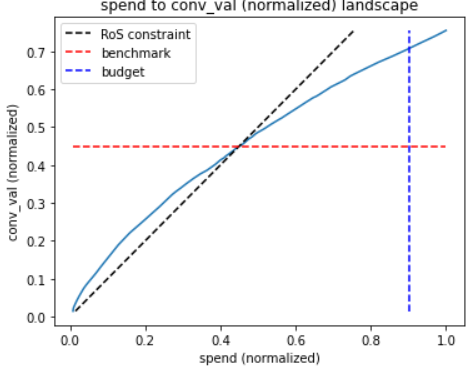}
        \caption{ROS constraint is binding.}
    \end{subfigure}
    \caption{The optimal operating points of an example campaign. The achievable curve (solid blue) delineates the pairs of spend-conversion value pairs that can be achieved by different bidding multipliers. The black diagonal dotted line captures the ROS constraint (feasible pairs should lie above this line), the blue vertical dotted line captures the budget constraint (feasible pairs should lie to the left of this line). The optimal operating point is the smallest of the intersection points of the achievable curve with one of the constraints and is shown using the red horizontal dotted line. In (a) the budget constraint is binding, while in (b) the ROS constraint is binding.}\label{fig:opt_exp}
  \end{center}
\end{figure}
\end{document}